\documentclass[a4paper,reqno]{amsart}

\usepackage[table,xcdraw,dvipsnames]{xcolor}

\usepackage{microtype}

\usepackage{bbm}
\usepackage{comment}
\usepackage{float}
\usepackage{colortbl}
\usepackage[T1]{fontenc}
\usepackage[margin=2cm]{geometry}
\usepackage[foot]{amsaddr}
\usepackage[colorlinks, linkcolor = blue, citecolor = blue, urlcolor = blue]{hyperref}

\usepackage{amsmath,amsthm,amssymb,amsfonts,mathtools}
\usepackage{ytableau}
\usepackage{pifont}

\usepackage[normalem]{ulem}

\usepackage{mathtools}

\usepackage{thmtools}

\usepackage{mathrsfs}

\usepackage{cleveref}

\usepackage[shortlabels]{enumitem}

\usepackage{tabularx}
\usepackage{makecell}
\usepackage{multirow}
\usepackage{hhline}
\usepackage{diagbox}
\usepackage{booktabs}
\usepackage{slashbox}
\usepackage{cases}

\usepackage{caption}

\usepackage[group-separator={,},group-minimum-digits=4]{siunitx}

\usepackage{epigraph} 

\usepackage{subcaption}  

\usepackage{graphicx}
\usepackage{tikz}
\usetikzlibrary{quantikz2}
\usetikzlibrary{backgrounds}
\usetikzlibrary{calc}
\usetikzlibrary{shapes}
\usetikzlibrary{decorations.pathreplacing}
\usetikzlibrary{external}
\usetikzlibrary{arrows.meta}

\usepackage{nicematrix}
\NiceMatrixOptions{cell-space-top-limit=5pt,cell-space-bottom-limit=5pt}
\usepackage{nicefrac}

\usepackage{ytableau,varwidth}

\usepackage[backend=biber,style=alphabetic,maxnames=9,maxalphanames=5,isbn=false,backref]{biblatex}

\usepackage{arydshln}

\usepackage{adjustbox}

\makeatletter
\renewcommand\section{\@startsection{section}{1}%
  \z@{.7\linespacing\@plus\linespacing}{.5\linespacing}%
  {\normalfont\Large\scshape\centering}}
\renewcommand\subsection{\@startsection{subsection}{2}%
  \z@{.5\linespacing\@plus.7\linespacing}{-.5em}%
  {\normalfont\large\bfseries}}
\makeatother

\newtheorem*{theorem*}{Theorem}
\newtheorem{theorem}{Theorem}

\newtheorem{proposition}[theorem]{Proposition}
\newtheorem{lemma}[theorem]{Lemma}
\newtheorem*{lemma*}{Lemma}
\newtheorem{definition}[theorem]{Definition}
\newtheorem{corollary}[theorem]{Corollary}
\newtheorem{remark}[theorem]{Remark}

\newtheorem{result}{Result}

\newcommand{\step}[2]{\par\medskip\textsc{Step #1} \textit{(#2).}\ }

\crefname{lemma}{Lemma}{Lemmas}
\crefname{definition}{Definition}{Definitions}
\crefname{theorem}{Theorem}{Theorems}
\crefname{conjecture}{Conjecture}{Conjectures}
\crefname{section}{Section}{Sections}
\crefname{claim}{Claim}{Claims}
\crefname{appendix}{Appendix}{Appendices}
\crefname{figure}{Fig.}{Figs.}
\crefname{table}{Table}{Tables}
\crefname{proposition}{Proposition}{Propositions}
\crefname{corollary}{Corollary}{Corollaries}
\crefname{example}{Example}{Examples}
\crefname{remark}{Remark}{Remarks}

\newcommand{\tr}{\operatorname{Tr}}
\newcommand{\spec}{\operatorname{spec}}

\newcommand{\TD}{T}              
\newcommand{\eps}{\varepsilon}

\newcommand{\pp}{\mathbf{p}}

\DeclarePairedDelimiterX\Set[1]\{\}{%
    
    #1
}

\newcommand{\defeq}{\vcentcolon=}

\renewcommand{\leq}{\leqslant}
\renewcommand{\geq}{\geqslant}

\renewcommand{\bra}[1]{\langle{#1}\rvert}
\renewcommand{\ket}[1]{\lvert{#1}\rangle}

\newcommand{\ketbra}[2]{\ket{#1}\bra{#2}}

\newcommand{\C}{\mathbb{C}} 
\newcommand{\R}{\mathbb{R}} 

\DeclareMathOperator{\poly}{poly} 

\DeclareMathOperator{\GL}{GL} 
\DeclareMathOperator{\End}{End} 
\DeclareMathOperator{\diag}{diag} 

\newcommand{\id}{\mathbbm{1}}

\newcommand\restr[2]{{
  \left.\kern-\nulldelimiterspace 
  #1 
  \vphantom{\big|} 
  \right|_{#2} 
  }}

\newcommand{\eat}[1]{}

\newcommand{\noisy}[1]{\textcolor{BrickRed}{#1}}
\newcommand{\clean}[1]{\textcolor{NavyBlue}{#1}}
\tikzset{wssgate/.style={fill=gray!12}}
\newcommand{\downarrowtext}[1]{%
  \begin{tikzpicture}[baseline=(current bounding box.center)]
    \draw[-{Stealth[length=2.2mm]}, thick] (0,0.45) -- (0,-0.45);
    \node[anchor=west, text width=8.4cm, font=\footnotesize, align=left] at (0.35,0) {#1};
  \end{tikzpicture}}

\title[Robustness of quantum spectrum estimation]{Robustness of quantum spectrum estimation: \\ weak Schur sampling under noisy inputs}

\author{Vladyslav Visnevskyi$^*$}
\author{Laura Mančinska$^\dagger$}
\dedicatory{Centre for the Mathematics of Quantum Theory (QMATH), University of Copenhagen}

\address[$*$]{vv@math.ku.dk}
\address[$\dagger$]{mancinska@math.ku.dk}

\begin{document}

\begin{abstract}
    We initiate the study of weak Schur sampling (WSS) under noisy inputs. Many optimal quantum learning algorithms rely on this measurement, for instance in quantum spectrum estimation (QSE) and quantum state tomography (QST). The standard analysis of WSS, QSE and QST assumes that the $n$ input copies of the target state $\rho$ are identical, i.e.\ $\rho^{\otimes n}$. We study what happens when they are not, a more realistic noisy scenario that breaks the very permutation symmetry on which WSS is built. In the simplest model the input is a noisy product state: the copies are independent but not necessarily identical and each is $\epsilon$-close in trace distance to the $d$-dimensional target state $\rho$. A naive argument via the data-processing inequality allows the per-copy errors to accumulate to~$n\epsilon$ on the measurement outcome. We show that for QSE via WSS this does not~happen: the expected estimation error in total variation is at most $\epsilon+\eta(n,d)$, where $\eta(n,d)=O(d/\sqrt n)$ is the noiseless error, and the additive noise term $\epsilon$ is optimal. The output error does not grow with $n$ even though the input error does; in this sense WSS is robust. Beyond noisy product states, for arbitrary input~$\omega$ we show that the error is at most $\eta(n,d)+\frac{1}{n}\|\omega-\rho^{\otimes n}\|_{W_1}$, where $W_1$ is the quantum Wasserstein distance of De~Palma~et~al. In particular, our results imply, for the first time, noise-robustness of Keyl and Werner's seminal algorithm for QSE, more than two decades after its introduction. Finally, we show that the noise cannot be too bad: within our model of robustness the promise on the input cannot be weakened to closeness of single-copy marginals or to a global trace-distance budget alone, as in both cases some inputs defeat every measurement and estimator.
\end{abstract}

\maketitle
\setcounter{tocdepth}{1}
\tableofcontents

\section{Introduction}
\label{sec:intro}

Quantum learning theory is a core subfield of quantum information science.
Its primary tasks include such fundamental primitives as quantum state
tomography, spectrum estimation, certification, hypothesis testing, and more.
Several of these tasks have been studied in the presence of noise, but when it comes to the \emph{sample-optimal} algorithms for those problems the picture is far less complete. For many quantum learning problems the optimal algorithms have only been analysed under perfect input. Studying optimal learning algorithms in the more realistic setting of noisy input is a necessary step towards making them practical.  

In this work we take a first step towards establishing the noise-robustness of optimal algorithms for a broad category of quantum learning problems. Specifically, we establish robustness of a core measurement in quantum learning called \textit{weak Schur sampling}, used in many learning algorithms, including sample-optimal ones
\cite{keyl2001estimating,bakshi2026optimalspectrumestimation,keyl2006quantum, odonnell2016efficient,haah2017sample,pelecanos2025debiased,pelecanos2025mixed,
odonnell2015quantum,acharya2020measuring,lowe2026random}. In particular, this measurement is used in optimal approaches to quantum spectrum estimation \cite{bakshi2026optimalspectrumestimation} and quantum state tomography \cite{odonnell2016efficient,haah2017sample}, two fundamental tasks in quantum learning. We introduce noise into the input to weak Schur sampling, prove that this measurement is robust under natural noise models, such as independent noise on each input copy, and show that it cannot be robust under the weaker promises that constrain only the overall trace distance to the ideal input or only the single-copy marginals. To the best of our knowledge, we are the first to analyse robustness of weak Schur sampling under noisy input.

Quantum spectrum estimation is the learning task in which weak Schur sampling plays the most direct role. The standard Keyl-Werner estimator \cite{keyl2001estimating}, which reads the spectrum straight off the measurement outcome, already achieves the collective $O(d^2/\delta^2)$ copy complexity for learning the spectrum of a state of local dimension $d$ with precision~$\delta$~\cite{odonnell2016efficient}. The sample-optimal algorithm of \cite{bakshi2026optimalspectrumestimation} improves on it only by polylogarithmic factors in $d$ and $\delta$, through additional steps built on the same core subroutine. 
Our results settle the robustness of the Keyl-Werner estimator, and hence that of the central step of the optimal algorithm. This reduces its robustness to that of the remaining steps of the algorithm; whether those steps are robust is one of the two main open questions of this work. The same holds for optimal full state tomography \cite{keyl2006quantum,odonnell2016efficient,haah2017sample}, where weak Schur sampling is followed by a second measurement that estimates the eigenbasis: our results settle the first stage, and the second is the other main open question.

\subsection*{Setup}

Let $\rho$ be an unknown density operator on $\C^d$. \emph{Quantum state
tomography} asks to reconstruct $\rho$ from measurements on $n$ copies
$\rho^{\otimes n}$ \cite{paris2004quantum,haah2017sample,odonnell2016efficient},
and \emph{spectrum estimation} asks for the ordered eigenvalue vector
\begin{equation*}
  \pp=\spec(\rho)=(p_1\ge p_2\ge\cdots\ge p_d\ge0),\qquad \sum_i p_i=1,
\end{equation*}

The sample-optimal algorithms for both full state tomography \cite{odonnell2016efficient,haah2017sample} and spectrum estimation \cite{bakshi2026optimalspectrumestimation} are \emph{collective}: they measure all $n$ copies jointly.  A core measurement
present in both optimal algorithms is the \textit{weak Schur
sampling}.

On the space $(\C^d)^{\otimes n}$ of $n$ qudits with local dimension $d$, weak Schur sampling outputs \emph{partitions} of $n$ into at most $d$ parts, that is, vectors $\lambda=(\lambda_1\ge\cdots\ge\lambda_d\ge0)$ of integers summing to $n$. On input $\rho^{\otimes n}$ of $n$ perfect copies of $\rho$, the normalised outcome $\lambda/n$ post-processed from weak Schur sampling concentrates around $\spec(\rho)$. This gives the empirical estimator $\hat\pp=\lambda/n$ of the spectrum, which reaches $\ell^1$ precision $\delta$ with $O(d^2/\delta^2)$ copies
\cite{keyl2001estimating,odonnell2016efficient}.

For full state tomography, the eigenbasis has to be estimated too. Keyl's measurement \cite{keyl2006quantum} completes weak Schur sampling into a full tomography protocol by a second, covariant measurement inside the
block $\lambda$ that estimates the eigenbasis. This protocol is sample-optimal for full state tomography, with $\Theta(d^2/\delta^2)$ copies in trace distance \cite{odonnell2016efficient,haah2017sample}.

\subsection*{Motivation of the problem}

The above-mentioned results give statements about the standard assumption of perfect-copies input $\rho^{\otimes n}$.
In an experiment the $n$ copies are produced by a device, and no device
produces the same state twice.  This is true of noisy intermediate-scale
hardware, where drift and imperfect calibration change the prepared state from shot to shot, and it remains true of fault-tolerant hardware: below-threshold error correction suppresses the logical error rate but does not eliminate it \cite{acharya2025quantum}, so the logical state delivered on each shot is a slightly and unpredictably wrong version of the intended one. State tomography and spectrum estimation are precisely the tools one would use to verify such a machine and to learn its residual errors
\cite{eisert2020quantum,harper2020efficient}, and to characterise the output of a noisy quantum simulator. In all of these applications one runs the estimator on copies that are only approximately identical, and it is not obvious that the guarantees proven for $\rho^{\otimes n}$ survive. A collective measurement in the optimal quantum spectrum estimation algorithm exploits the permutation symmetry of $\rho^{\otimes n}$, while a noisy asymmetric input $\rho_1\otimes\cdots\otimes\rho_n$ has no permutation symmetry whatsoever. 

From a naive data-processing argument, one might fear that $n$ small errors on each copy accumulate to an error of order
$n\eps$ in the collective statistics, growing with system size. Under perfect input, precision $\delta$ requires $n = O(d^2/\delta^2)$ copies, therefore such an accumulated error would already exceed
the target precision for $\eps = \Omega(\delta^3/d^2)$, and would be of order one for $\eps = \Omega(\delta^2/d^2)$, rendering the estimate too noisy to be useful. Effectively, an output noise growing with system size $n$ would put an upper bound on both the rank of the states whose spectrum is possible to estimate, and on the achievable precision. An
arbitrarily small per-copy noise level would then ruin collective protocols in
exactly the many-copy measurement regime in which they are supposed to outperform
single-copy measurement ones, whereas a robust single-copy protocol run on the same copies simply pays the extra price $\eps$ in precision \cite{stilckfranca2021fast}.

\subsection{Our results}
\hfill
\\
We take a first step towards solving robustness of optimal quantum spectrum estimation \cite{bakshi2026optimalspectrumestimation}, optimal full state tomography \cite{odonnell2016efficient,haah2017sample}, as well as lots of other learning tasks
\cite{acharya2020measuring, odonnell2015quantum, lowe2026random} by
analyzing weak Schur sampling, a core step present in their algorithms. To present the ideas in the simplest setting we first work with a \emph{product} noise model: the noisy input to weak Schur sampling is
\begin{equation*}
  \mathcal{E}(\rho^{\otimes n})
  =\mathcal{E}_1(\rho)\otimes\cdots\otimes\mathcal{E}_n(\rho),
  \qquad
  \TD\bigl(\mathcal{E}_i(\rho),\rho\bigr)\le\eps\quad\forall i,
\end{equation*}
for arbitrary, possibly different, quantum channels $\mathcal{E}_i$; the copies
are independent but not identically distributed. Our first main result
(\cref{thm:robust}) reads:
\begin{result}[Weak Schur sampling under product noise, \cref{thm:robust,prop:floor}]
\label{res:product}
  Consider the product noise model on $n$ input copies, in which each noisy
  copy is within trace distance $\eps$ of the target state $\rho$ on $\C^d$.
  On such input weak Schur sampling with the estimator $\hat\pp=\lambda/n$ learns the
  spectrum $\pp$ of $\rho$ with expected error
  \begin{equation*}
    \mathbb{E}\bigl[\tfrac12\|\hat\pp-\pp\|_1\bigr]
    \;\le\;\eps+\eta(n,d),
    \qquad \eta(n,d)=O\!\bigl(d/\sqrt n\bigr),
  \end{equation*}
  where $\eta(n,d)$ is the error of the same estimator on noiseless input $\rho^{\otimes n}$.
  That is, $n=O(d^2/\delta^2)$ noisy copies give output precision
  $\delta+\eps$, and the noise-induced output error $\eps$ cannot be avoided by any measurement and estimator.
\end{result}
Thus the noise costs an additive term $\eps$ and nothing else. We call it the \emph{noise floor}: a noise-induced part of the output error that cannot be reduced by taking more copies. The error decreases at the noiseless rate until it reaches the floor, and the collective $d^2$ sample complexity is preserved. The floor is optimal (\cref{prop:floor}): $n$~identical copies of a single $\eps$-perturbed state are an example of an input allowed by our noise model, and on them no estimator whatsoever can do better than $\Omega(\eps)$. In other words, our results showcase that even for a noise model with input error growing with system size $n$, the noise-induced output error of weak Schur sampling spectrum estimator does not grow with the system size. In this sense, it is robust to noisy input, and omits the naive concerns of the output error scaling with $n$.

The main idea of the proof is straightforward: we start replacing the noisy registers $\mathcal{E}_i(\rho)$ in $\bigotimes_i \mathcal{E}_i(\rho)$ one at a time by the noiseless state $\rho$, and iteratively construct a chain of inequalities (see \cref{fig:hybrid} for an illustration). The expectation of any function $f$ of the outcome is a linear functional $\tr[\mathcal{O}_f\,\omega]$ of the input $\omega$, so one can bound each replacement separately.  The cost of replacing register $i$ is
$\tr[(\mathcal{E}_i(\rho)-\rho)\,h_i]$, where $h_i$ is the effective $d\times d$
observable that the collective measurement presents to register $i$ when the
other registers are fixed.  Since $\mathcal{E}_i(\rho)-\rho$ is traceless with trace norm at most $2\eps$, only the \emph{eigenvalue spread} of $h_i$ matters (\cref{lem:centering}), and for the error of the spectrum estimator this spread is at most~$1/n$~(\cref{lem:spread}).
The $n$ replacements therefore cost at most $n\cdot\eps\cdot(1/n)=\eps$ in total, uniformly in $n$, and what remains is the noiseless protocol on $\rho^{\otimes n}$. The proof uses no information about the shape of the outcome distribution, and only a very limited amount of representation theory (nothing beyond the branching rule for $S_{n-1}\subset S_n$), so we think one of its main virtues is its simplicity.


\begin{figure}[H]
  \centering
  \begin{quantikz}[row sep={0.5cm,between origins}, column sep=0.4cm, wire types={q,q,n,q}]
    \lstick{\noisy{$\sigma_1$}} & \gate[4,style=wssgate]{\text{WSS}} &            & \rstick[4]{$\lambda\sim P(\lambda;\omega^{(0)})$} \\
    \lstick{\noisy{$\sigma_2$}} &                                     &            & \\
    \lstick{$\vdots$}           &                                     & \ \vdots\  & \\
    \lstick{\noisy{$\sigma_n$}} &                                     &            &
  \end{quantikz}
  \par\medskip
  \downarrowtext{Replace register $1$ by $\rho$: the expected
    total-variation error $\tfrac12\|\lambda/n-\spec(\rho)\|_1$ under
    $P(\lambda;\omega^{(0)})$ changes by at most $\eps/n$ when
    $\omega^{(0)}\rightarrow\omega^{(1)}$ (\cref{lem:centering,lem:spread}).}
  \par\medskip
  \begin{quantikz}[row sep={0.5cm,between origins}, column sep=0.4cm, wire types={q,q,n,q}]
    \lstick{\clean{$\rho$}}     & \gate[4,style=wssgate]{\text{WSS}} &            & \rstick[4]{$\lambda\sim P(\lambda;\omega^{(1)})$} \\
    \lstick{\noisy{$\sigma_2$}} &                                     &            & \\
    \lstick{$\vdots$}           &                                     & \ \vdots\  & \\
    \lstick{\noisy{$\sigma_n$}} &                                     &            &
  \end{quantikz}
  \par\medskip
  \downarrowtext{$\dots$repeat one register at a time, $n$ replacements in
    total, thus total change in cost at most $n\cdot\eps/n=\eps$.  After $n$
    replacements the final state is perfect copies $\rho^{\otimes n}$, so
    final output noise is change~$+$~noise~on~perfect: $\epsilon+\eta(n,d)$.}
  \par\medskip
  \begin{quantikz}[row sep={0.5cm,between origins}, column sep=0.4cm, wire types={q,q,n,q}]
    \lstick{\clean{$\rho$}}     & \gate[4,style=wssgate]{\text{WSS}} &            & \rstick[4]{$\lambda\sim P(\lambda;\rho^{\otimes n})$} \\
    \lstick{\clean{$\rho$}}     &                                     &            & \\
    \lstick{$\vdots$}           &                                     & \ \vdots\  & \\
    \lstick{\clean{$\rho$}}     &                                     &            &
  \end{quantikz}
  
  \caption{The proof technique behind Result~\ref{res:product}, WSS is weak Schur sampling.  Starting from the noisy product input $\omega^{(0)}=\sigma_1\otimes\cdots\otimes\sigma_n$ (red), the
  noisy registers are replaced by the ideal copy $\rho$ (blue) one at a time,
  through intermediate states $\omega^{(k)}=\rho^{\otimes k}\otimes\sigma_{k+1}\otimes\cdots\otimes\sigma_n$,
  until the noiseless input $\rho^{\otimes n}$ is reached, on which WSS has error $\eta(n,d)$. We show that each replacement changes the expected total variation error by at most $\eps/n$, so the $n$ replacements cost $\eps$ in total.  The second main ingredient is that it does not matter in which order the registers are replaced, nor which register is replaced at a given step: WSS commutes with all permutations of the registers (\cref{lem:perminv}), so every step is the same single-register bound.}\label{fig:hybrid}
\end{figure}
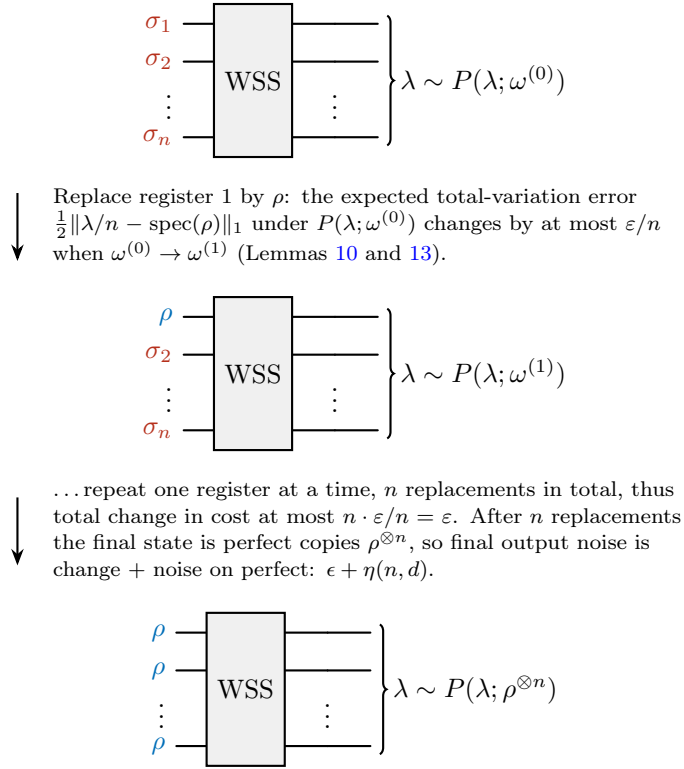

We further generalize our results beyond the product noise on input. For the replacement step, we apply the argument in De~Palma, Marvian, Trevisan and Lloyd~\cite{depalma2021wasserstein} behind the duality
between the quantum Wasserstein distance of order one and the quantum
Lipschitz constant. The second main result reads:

\begin{result}[Weak Schur sampling under arbitrary noise, \cref{thm:wasserstein,cor:w1cases}]
  Let $\omega$ be an arbitrary, possibly correlated or entangled, state of
  $n$ input copies, and let $\|\omega-\rho^{\otimes n}\|_{W_1}$
  be its distance from the ideal input $\rho^{\otimes n}$ in the
  quantum Wasserstein distance of order one. On such input weak Schur
  sampling with the estimator $\hat\pp=\lambda/n$ learns the spectrum $\pp$ of $\rho$ with expected error
  \begin{equation*}
    \mathbb{E}\bigl[\tfrac12\|\hat\pp-\pp\|_1\bigr]
    \;\le\; \tfrac1n \|\omega-\rho^{\otimes n}\|_{W_1} + \eta(n,d),
    \qquad \eta(n,d)=O\!\bigl(d/\sqrt n\bigr),
  \end{equation*}
  with the same noiseless error $\eta(n,d)$ as above.
\end{result}

In our above-mentioned simplest setting, product $\eps$-noise is the case $\|\omega-\rho^{\otimes n}\|_{W_1}\le n\eps$. In the generalized setting:
a few completely corrupted registers among perfect ones, correlated noise of
small Wasserstein size, and the Wasserstein almost-i.i.d.\ sources of
\cite{girardi2026new} are all covered by the same bound.  The concentration statements that accompany the expectation bound, single-shot guarantees,
confidence amplification, and per-realisation tail bounds are
collected in \cref{app:concentration}.

To introduce a bit of intuition behind the second result, we can see the quantum Wasserstein distance of order one as a quantum extension of
the Hamming distance between strings: on computational basis states it counts
the number of registers in which two strings differ, so that $\ket{0000}$ is
at distance $1$ from $\ket{0001}$ but at distance $4$ from $\ket{1111}$,
whereas the trace distance is $1$ in both cases.  This is exactly the kind of
bookkeeping one needs when the registers are replaced one at a time, like in our proof technique, and it
makes $W_1$ the natural distance for our purposes.
Beyond just realizing the applicability of \cite{depalma2021wasserstein} to our problem, our novel contribution is the input to that argument, \cref{lem:spread}, which in the language of \cite{depalma2021wasserstein} says that the observable associated with the error of the WSS spectrum estimator has quantum Lipschitz constant at most $1/n$.  Once this is established, applying the duality in \cite{depalma2021wasserstein} gives the result for \emph{arbitrary} inputs $\omega$, correlated or entangled (\cref{thm:wasserstein}).

Importantly, our robustness guarantees in \cref{thm:robust} and \cref{thm:wasserstein} are not specific to the total-variation error $\tfrac12\|\hat\pp-\pp\|_1 = \tfrac12\|\lambda/n-\pp\|_1 = f(\lambda)$, viewed as a function of $\lambda$. More generally: for \emph{any} function $f$ of the outcome, the expectation of $f(\lambda)$ moves by at most $\Delta(f)\,\|\omega-\rho^{\otimes n}\|_{W_1}$ (see \cref{thm:wasserstein}), and therefore by at most $n\eps\,\Delta(f)$ under product noise, where $\Delta(f)$ is the largest change $|f(\lambda)-f(\lambda')|$ over all pairs of partitions $\lambda,\lambda'$ that differ by moving a single unit from one entry to another, i.e.\ $\lambda'=\lambda-e_i+e_j=(\lambda_1,\dots,\lambda_i-1,\dots,\lambda_j+1,\dots,\lambda_d)$ (\cref{def:oneboxvar}). Every outcome statistic with $\Delta(f)=O(1/n)$ therefore inherits the same $n$-independent noise floor under product noise, just like the total variation distance.

Finally we show what happens when the promise on the input to weak Schur sampling is too weak, i.e.\ the noise model is too strong.
\begin{result}[Failure under weaker input restrictions, \cref{prop:marginalfailure,prop:globalbudget}]
  Neither of the following two weaker promises on the input to quantum spectrum estimation suffices for robustness, even for the best possible measurement and estimator.
  \begin{enumerate}
    \item[(i)] \textup{(Marginals only.)}  If the input is only promised
      to have every single-copy marginal within trace distance $\eps$ of the target $\rho$, then already at $\eps=0$ there is an input, all of whose marginals are \emph{exactly} $\rho$, on which the estimator
      $\hat\pp=\lambda/n$ has constant error for every $n$, and no
      measurement and estimator can reduce the worst-case error below a
      constant, independent of $n$.
    \item[(ii)] \textup{(Global trace distance only.)}  If the input is
      only promised to be within trace distance $\theta$ of
      $\rho^{\otimes n}$, the worst-case error is $\Theta(\min(\theta,1))$, for weak Schur sampling and for every measurement and estimator alike. For an input noise $\theta$ naturally growing with the system size $n$, as we should expect, we therefore guarantee that under this noise model the output error grows with $n$ as well.
  \end{enumerate}
\end{result}

These examples do not imply that the product and Wasserstein models, under which we prove robustness, are the only noise models that admit it. They show that the noise cannot be arbitrarily bad, and they leave a large space of intermediate noise models and alternative notions of robustness open. We discuss both in \cref{sec:open}.

\subsection{Applications}
\hfill
\\
As mentioned above, our noise robustness result for weak Schur sampling gives, for the first time, a robustness guarantee for the standard Keyl–Werner algorithm for quantum spectrum estimation~\cite{keyl2001estimating}. Beyond this, our results have major applicability to a number of other robustness questions in quantum information theory.

\textit{Full state tomography and purification.}
The sample-optimal tomography protocols
\cite{keyl2006quantum,odonnell2016efficient,haah2017sample}, including the
unbiased estimator of Pelecanos, Spilecki and Wright
\cite{pelecanos2025debiased}, perform weak Schur sampling followed by an
eigenbasis measurement inside the block $\lambda$ (\cref{app:keyl}).
Tomography via random purification~\cite{pelecanos2025mixed,tang2025conjugate},
optimal purification \cite{cirac1999optimal,keyl2001rate} and purity
amplification \cite{li2025optimal} likewise begin with weak Schur sampling and
then act on the post-measurement state.  For all of these our results settle the first stage, the robustness of weak Schur sampling, and with it the spectral part of a tomographic estimate, is robust with the optimal noise error floor. The robustness of the second stage consisting of the eigenbasis measurement remains open (\cref{sec:open}, item~1).

\textit{Random dimension reduction.}
Lowe and Tan \cite{lowe2026random} show that, for input states of rank at most $r$, the $n$-copy input on $(\C^d)^{\otimes n}$ can be mapped to
$(\C^r)^{\otimes n}$ while keeping the information needed to learn properties invariant under a common unitary rotation, such as the fidelity of two states, so that the copy complexity of such tasks depends on $r$ rather than on $d$. Their channel performs weak Schur sampling, outputs a fixed state if $\ell(\lambda)>r$, and otherwise re-embeds the post-measurement state in local dimension $r$. They ask \cite[Sec.~1.2]{lowe2026random} whether the procedure is robust to inputs that are only close to low rank, by analogy with \cite[Cor.~1.6]{badescu2019quantum}.  Identical copies of a state $\sigma$ that is $\eps$-close to a state of rank at most $r$ are an instance of the product noise model, so our robustness guarantees for weak Schur sampling could be of use for robustness analysis of random dimension reduction.

\textit{Learning the average of non-identical states.}
De~Palma, Fanizza, Mowry and O'Donnell \cite{depalma2025noniid} leave open
whether the average state $\rho_{\mathrm{avg}}=\frac1n\sum_i\sigma_i$ of a
product input $\omega=\sigma_1\otimes\cdots\otimes\sigma_n$ can be learned with the i.i.d.\ copy complexity. For its spectrum, applying our \cref{cor:w1cases}(i) with $\rho=\rho_{\mathrm{avg}}$ immediately gives
\begin{equation*}
  \mathbb{E}^{\,\omega}_{\mathrm{meas}}
    \bigl[\tfrac12\|\hat\pp-\spec(\rho_{\mathrm{avg}})\|_1\bigr]
  \;\le\;\eta(n,d)+\frac1n\sum_{i=1}^n\TD(\sigma_i,\rho_{\mathrm{avg}}),
\end{equation*}
a useful guarantee when the copies cluster around their average, as for a
device whose state drifts slowly.

\textit{Functionals of the spectrum and tolerant testing.} For the type of problems that ask questions about some functions of the spectrum, our analysis applies immediately.
As previously mentioned in our results, the robustness guarantees that we provide apply not just to total variation distance, but also to other functions $f(\lambda)$ of the isotypic block, with the bounds depending on the property of these functions that we call one-box variation $\Delta(f)$. 
For example, estimators of the entropy \cite{acharya2020measuring}, of the purity, and of the weight of the largest eigenvalues \cite{odonnell2016efficient,odonnell2017efficient}, as well as the spectrum testers of O'Donnell and Wright \cite{odonnell2015quantum}, are all functions of $\lambda$, so our analysis reduces their robustness question to computing the one-box variation. It's easy to verify that per our definitions the partial sums $\sum_{i\le k}\lambda_i/n$ and the plug-in purity $\sum_i(\lambda_i/n)^2$ have one-box variation at most $1/n$ and $2/n$, giving floors $\eps$ and $2\eps$.  For the plug-in entropy $S(\lambda/n)$, \cref{thm:robust} and the continuity bound of Audenaert \cite{audenaert2007sharp} give an expected error of at most $(\eps+\eta)\log(d-1)+h(\eps+\eta)$, with $h$ the binary entropy and $\eta=\eta(n,d)$, whenever $\eps+\eta\le\frac12$.

\textit{Adversarially corrupted copies.}
Aliakbarpour, Braverman, Chia and Liu \cite{aliakbarpour2025adversarially}
study state learning and testing when an adversary corrupts a
$\gamma$-fraction of the data, and show that with non-adaptive single-copy
measurements the error of learning a general state must grow with the
dimension, as $\gamma\sqrt d$ in the worst case.  For collective spectrum
estimation, an input in which $\gamma n$ registers are replaced by an
arbitrary, possibly entangled, state is covered by \cref{cor:w1cases}(ii),
with the dimension-free floor $2\gamma$.  The tasks differ, so this is not a separation, but it shows that for the spectrum the collective measurement tolerates such corruption at a cost independent of the dimension. However, proving robustness of the eigenbasis measurement together with our result would imply robustness of full state tomography, which then would indeed give a separation between the noise floor for single-copy measurements and collective measurements.

\subsection{Connections to existing work}
\hfill

\textit{Robustness of single-copy protocols.}
Noise robustness of tomography has mostly been studied for
protocols in which each copy is measured separately, and the existing
works differ in what is assumed to be noisy.  Closest to our model are those in which the \emph{prepared input states} vary. Stilck Fran\c{c}a, Brand\~ao and Kueng \cite{stilckfranca2021fast} give a fast tomography algorithm from few measurement bases whose guarantees hold under bounded state-preparation and measurement errors, the prepared state being allowed to differ from one measurement setting to the next; Zambrano \cite{zambrano2026tomography} shows that projected least-squares tomography from non-adaptive single-copy measurements recovers the \emph{time-averaged} state at the i.i.d.\ rate even
when the source is fully adaptive (in other words, for an adversarial
source); and Gu\c{t}\u{a}, Kahn, Kueng and Tropp \cite{guta2020fast}, who
introduced projected least squares, prove robustness to approximately
low-rank targets and remark that their analysis should extend to drifting
state preparation.  A second group treats the \emph{measurement apparatus} as the noisy component: Chen, Yu, Zeng and Flammia \cite{chen2021robust} and Koh
and Grewal \cite{koh2022classical} make the classical shadows protocol
\cite{huang2020predicting} robust to noisy measurements, and Jayakumar et
al.\ \cite{jayakumar2024universal} learn the state and the detector noise
simultaneously.  A third group considers \emph{adversarial corruption}:
Aliakbarpour, Braverman, Chia and Liu \cite{aliakbarpour2025adversarially}
let an adversary alter a $\gamma$-fraction of the measurement outcomes and
show that the resulting error floor of non-adaptive single-copy learning must grow with the dimension, as $\gamma\sqrt d$ in the worst case.  Finally, Li et al.~\cite{li2026efficient} study state and process tomography of noisy circuits; there the noise sits inside the preparation and is part of the target, and the copies remain identical.  Their states are outputs of random logarithmic-depth circuits with constant-strength noise, which damps all high-weight Pauli components, so only polynomially many low-weight coefficients need to be learned.  In all of these the estimate is a classical function of $n$ independent single-copy outcomes, and robustness can be traced back to the fact that an average of $n$ independent quantities is perturbed by at most the average perturbation. For a collective measurement, like in our case with weak Schur sampling, no such argument is available.

\textit{Learning from non-i.i.d.\ sources.} A second line of work considers tasks that assume independent identically distributed input (perfect copies) and drops the i.i.d.\ assumption altogether.  Christandl and Renner \cite{christandl2012reliable} construct confidence regions for tomography that are valid for arbitrary, possibly correlated, inputs, and Fawzi, Kueng, Markham and Oufkir \cite{fawzi2024learning} give a general reduction that adapts \emph{any} i.i.d.\ learning algorithm to arbitrary inputs, beyond just state tomography and spectrum estimation. Both go through de~Finetti-type arguments, apply to a randomly chosen subset of the systems, and pay a polynomial factor in the dimension for entangled measurements. This polynomial price is significant: for arbitrary, possibly entangled, measurements, the reduction of \cite{fawzi2024learning} emulates an i.i.d.\ algorithm that uses $k$ copies with $N = O(k^3 d^2 \log d)$ systems, up to factors depending on the accuracy and failure probability, and with $N = O\bigl(k^2 \log d \,\log^2(k/\beta)/(\beta^2\eps^2)\bigr)$ systems for non-adaptive single-copy measurements, where $\beta$ is the failure probability.  Instantiated with sample-optimal collective
tomography, $k = O(d^2/\delta^2)$, the first bound is polynomial of much
higher degree than the $O(d^3/\delta^2)$ copies that suffice for tomography with single-copy measurements \cite{haah2017sample}, so the advantage of collective measurements over single-copy ones does not survive the reduction.  Moreover, the resulting guarantee concerns a prediction for a randomly held-out system rather than a fixed target state. Fanizza, Quek and Rosati \cite{fanizza2024learning} extend shadow tomography to non-identical product states, Fanizza, Salvia and Giovannetti \cite{fanizza2023testing} test identity of collections of states, and Go\v{c}anin, \v{S}upi\'c and Daki\'c \cite{gocanin2022sample} certify states device-independently from product sources. Closest to the setup described in our paper, De~Palma, Fanizza, Mowry and O'Donnell \cite{depalma2025noniid} study hypothesis testing on $\rho_1\otimes\cdots\otimes\rho_T$ with a collective measurement and prove a quantum Efron--Stein inequality for product states along the way. Their statistic is a sum of two-body swap operators, whose expectation on a product state is explicit by linearity; they leave open whether the average state $\rho_{\mathrm{avg}}=\frac1T\sum_i\rho_i$ can be learned with the i.i.d.\ copy complexity.  However, despite the similarities of their setup to our product noise model, as well as their usage of collective measurements, our problem is different in kind: the estimand is the spectrum of a fixed target $\rho$ to which every copy is promised to be close, the protocol is fixed in advance, and the estimator $\lambda/n$ is not a few-body statistic since it is constructed from all copies collectively.

\textit{Almost-i.i.d.\ sources and the quantum Wasserstein distance.}
Which kind of departures from $\rho^{\otimes n}$ preserve the predictions of the
i.i.d.\ theory has recently been asked systematically in quantum Shannon
theory.  Mazzola, Sutter and Renner \cite{mazzola2026almostiid} introduce
almost-i.i.d.\ states via a de~Finetti-type construction; Girardi, De~Palma
and Lami \cite{girardi2026new} add two further notions --- \emph{Wasserstein}
almost-i.i.d.\ sources, for which the quantum Wasserstein distance of order
one \cite{depalma2021wasserstein} to $\rho^{\otimes n}$ is $o(n)$, and
\emph{weakly} almost-i.i.d.\ sources, for which average $k$-body marginals
converge to $\rho^{\otimes k}$ --- and prove a strict hierarchy between the three; Girardi, Datta, De~Palma and Lami \cite{girardi2026robust} determine which notions preserve error exponents, compression rates and capacities. Estimation problems have not been treated in that programme, and the present paper can be read as its spectrum-estimation entry: weak Schur sampling is robust for Wasserstein almost-i.i.d.\ sources (\cref{thm:wasserstein}), and it fails for weakly almost-i.i.d.\ sources already at $k=1$ (\cref{prop:marginalfailure}). The quantum Wasserstein distance and its dual Lipschitz constant are also the natural language for our proof technique, as explained next.

\Cref{tab:comparison} summarises the works discussed in this subsection that are directly comparable to ours. The papers from the subsection that are not present in the table are either only tangentially related to our question, solve a different problem, assume a simplified class of states on input, or otherwise admit no direct comparison to our work. Sample complexities are given to leading order in target dimension $d$ and in the precision $\delta$ (rank $r$ where the work is rank-sensitive), at constant failure probability. Fawzi et al.\ \cite{fawzi2024learning} sample-complexity is computed by taking $k = O(d^2/\delta^2)$ for state tomography with collective measurements, and $k=O(d^3/\delta^2)$ for single-copy measurements, and substituting into the reduction formulas as described above. 

\begin{table}[H]
  \centering
  \footnotesize
  \setlength{\tabcolsep}{4pt}
  \renewcommand{\arraystretch}{1.15}
  \setlength{\dashlinedash}{1.5pt}
  \setlength{\dashlinegap}{2pt}
  \begin{tabularx}{\textwidth}{@{}
      >{\raggedright\arraybackslash}p{2.9cm}
      >{\raggedright\arraybackslash}p{1.8cm}
      >{\raggedright\arraybackslash}X
      >{\raggedright\arraybackslash}p{2.3cm}
      >{\raggedright\arraybackslash}p{4.8cm}@{}}
    \toprule
    Paper & Problem & Noise model & Measurements & Copies \\
    \midrule
    \textbf{This work}
      & QSE & product $\eps$-noise; $W_1$ budget
      & collective (WSS) & $O\bigl(d^2/(\delta-\eps)^2\bigr)$; floor $\eps$ optimal\\
    \midrule
    Stilck Fran\c{c}a et al.\ \cite{stilckfranca2021fast}
      & QST & bounded state-preparation and measurement error
      & single-copy & $O(dr^2\log d/\delta^4)$; floor $\Theta(\eps)$\\
    \hdashline
    Zambrano \cite{zambrano2026tomography}
      & QST & adaptive (adversarial) source
      & single-copy, non-adaptive
      & $O(dr^2\log d/\delta^2)$ (2-designs); \newline
        $O(d^{\log_2 3}\,r^2\log d/\delta^2)$ (Pauli bases) \\
    \hdashline
    Gu\c{t}\u{a} et al.\ \cite{guta2020fast}
      & QST & none (perfect copies), target only approximately rank $r$
      & single-copy & $O(dr^2\log d/\delta^2)$ \\
    \hdashline
    Aliakbarpour et al.\ \cite{aliakbarpour2025adversarially}
      & QST, testing & $\gamma$-fraction of outcomes corrupted
      & single-copy, non-adaptive
      & $d^{O(\log(1/\gamma))}$; floor $\tilde\Theta(\gamma\sqrt r)$ \\
    \hdashline
    Fawzi et al.\ \cite{fawzi2024learning}
      & any i.i.d.\ task & arbitrary input
      & any
      & For QST: $\tilde O(d^6/\delta^6)$ (single-copy), \newline
        $O(d^8\log d/\delta^8)$ (collective) \\
    \bottomrule
  \end{tabularx}
  \caption{Robustness results for quantum learning compared with the present work.  QST: state tomography; QSE: spectrum estimation; WSS: weak Schur sampling; $d$: target local dimension; $r$: target rank for $r \leq d$; $\delta$: output precision.}
  \label{tab:comparison}
\end{table}

\subsection*{Organisation}

\Cref{sec:prelim} recalls Schur--Weyl duality, weak Schur sampling, and the
noiseless quantum spectrum estimation; \cref{sec:setup} sets up the product noise model and
defines robustness; \cref{sec:main} contains the main result for product
noise and its proof; \cref{sec:wasserstein} recalls the quantum Wasserstein distance and generalises our result to arbitrary inputs;
\cref{sec:beyondproduct} treats the two relaxed noise classes on which
estimation fails; \cref{sec:open} discusses the scope of the results and
lists open questions. In the Appendix, \Cref{app:keyl} descibes the precision of the Keyl measurement for spectrum estimation, and also recalls the full Keyl measurement and the Keyl formula for the outcome distribution, and \cref{app:concentration} contains the single-shot, amplification and concentration statements.

\section{Preliminaries}
\label{sec:prelim}

\subsection{Notation}

A \emph{quantum state} (often shortened to \textit{state} in this work) on $\C^d$ is a positive semidefinite operator of unit trace.
The \emph{trace distance} between states $\sigma,\rho$ is
$\TD(\sigma,\rho)\defeq\tfrac12\|\sigma-\rho\|_1$, where $\|\cdot\|_1$ is the
trace norm; it is jointly convex and monotone under quantum channels
\cite{nielsen2000}.  For a state $\rho$ we write $\spec(\rho)\in\R^d$ for its
eigenvalues in non-increasing order, and for probability vectors
$q,q'\in\R^d$ we measure distance by $\tfrac12\|q-q'\|_1$, the total
variation distance.   

A \emph{partition} $\lambda\vdash n$ with at most $d$ parts is a vector
$(\lambda_1\ge\cdots\ge\lambda_d\ge0)$ of integers summing to $n$, and
$\ell(\lambda)$ is the number of non-zero parts.  We identify $\lambda$ with
its \emph{Young diagram}: a left-aligned array of $n$ unit squares, called
\emph{boxes}, with $\lambda_j$ boxes in row $j$, rows being stacked from the
top and weakly decreasing in length.
Thus~$\lambda=(3,1)\vdash4$ is $\raisebox{0.3\height}{\scalebox{0.6}{\ydiagram{3,1}}}$, and \emph{adding a box} to $\lambda$ means increasing one of the components $\lambda_i$ by one, such that it preserves the weakly decreasing ordering. Writing $e_j$ for the $j$-th standard basis vector, $\lambda+e_j$ is the diagram with one more box in row $j$, provided the result is again a partition, i.e.\ $\lambda_j+1\le\lambda_{j-1}$ (or $j=1$)
and $j\le d$.  For $\lambda=(3,1)$ and $d\ge3$ the rows $1$, $2$ and $3$ are
all addable,
\begin{equation*}
  \lambda+e_1=\ydiagram{4,1}\,,\qquad
  \lambda+e_2=\ydiagram{3,2}\,,\qquad
  \lambda+e_3=\ydiagram{3,1,1}\,,
\end{equation*}
whereas for $\lambda=(2,2)$ row $2$ is not addable, since $(2,3)$ is not a
partition. We write $\operatorname{add}(\mu)\subseteq[d]$ for the set of addable rows of $\mu$, i.e.\ the rows $j$ for which $\mu+e_j$ is again a partition with at most $d$ parts; thus $\operatorname{add}((3,1))=\{1,2,3\}$ for $d\ge3$ and $\operatorname{add}((2,2))=\{1,3\}$.

Every expectation in the paper carries a subscript naming the randomness it averages over, namely $\mathbb{E}_{\mathrm{meas}}$ averages over different iterations of the protocol (i.e.\ over shots, i.e.\ over final measurements instances), and $\mathbb{E}_{\mathrm{noise}}$ averages over different noise initializations.

\subsection{Schur--Weyl duality}
\label{subsec:sw}

The space $(\C^d)^{\otimes n}$, which we think of as $n$ qudits of local
dimension $d$, carries two commuting group actions.  The general linear group
$\GL(d,\C)$ acts \emph{diagonally}, that is, the same matrix acts on every
tensor factor,
\begin{equation}
  A\;\longmapsto\;A^{\otimes n},\qquad A\in\GL(d,\C),
  \label{eq:glaction}
\end{equation}
and the symmetric group $S_n$ acts by permuting the factors,
$\pi\,(\ket{\psi_1}\otimes\cdots\otimes\ket{\psi_n})
=\ket{\psi_{\pi^{-1}(1)}}\otimes\cdots\otimes\ket{\psi_{\pi^{-1}(n)}}$, which
physically is a sequence of qudit swaps.  The two actions commute, since
permuting the factors of $A^{\otimes n}$ returns $A^{\otimes n}$, and in fact
each of them generates the full commutant of the other in
$\End((\C^d)^{\otimes n})$ \cite{fulton1991representation,harrow2005applications}.
\emph{Schur--Weyl duality} is a theorem in representation theory that states the two actions simultaneously block diagonalize, because their actions generate full commutants of one another. The block diagonalization is:
\begin{equation}
  (\C^d)^{\otimes n}
  \;\cong\;
  \bigoplus_{\substack{\lambda\vdash n\\ \ell(\lambda)\le d}}
  V_\lambda\otimes M_\lambda ,
  \label{eq:sw}
\end{equation}
where the blocks, also called \emph{isotypic components} of the $\GL(d,\C)$ and $S_n$ actions, are labelled by the partitions $\lambda$ of $n$ with at most $d$ parts, i.e.\ Young diagrams with a total of $n$ boxes in $d$ rows. $V_\lambda$ carries an irreducible representation $\pi_\lambda$ of $\GL(d,\C)$, $M_\lambda$ carries an irreducible representation of $S_n$, and
different labels $\lambda$ give inequivalent representations of both groups.
Concretely, in a suitable basis every operator $A^{\otimes n}$ is block
diagonal, acting on the block $\lambda$ as $\pi_\lambda(A)\otimes\id_{M_\lambda}$,
and every permutation is block diagonal too, acting on the block $\lambda$
within the factor $M_\lambda$ only.  

The main fact about the labels $\lambda$ that we
use in this paper is that they are exactly the Young diagrams introduced above. Since the state of $n$ perfect input copies
$\rho^{\otimes n}$ is of the form \cref{eq:glaction} (it lies in the algebra
spanned by the group elements $A^{\otimes n}$), it is block diagonal with respect to
\cref{eq:sw} as well.  This is the starting point of weak Schur sampling,
defined next: because of the block structure, the orthogonal projectors
$\Pi_\lambda$ onto the blocks form a projective measurement, and since
$\rho^{\otimes n}$ is block diagonal, this measurement leaves the actual state on the registers within
each block untouched and only reveals the label $\lambda$ of the block.
\\

\paragraph{\textbf{The branching rule.}}
In our proofs we utilize the idea of perturbing only one register at a time, with permutation symmetry retained on the remaining $n-1$ registers. In this sense, we start from $n$ registers with a symmetry and reduce this symmetry to $n-1$ registers. One key ingredient we need from representation theory to describe this formally concerns the subgroup $S_{n-1}\subset S_n$ that permutes the first $n-1$ registers only. Restricted to $S_{n-1}$, the representation $M_\lambda$ of $S_n$ splits as
\begin{equation}
  M_\lambda\big|_{S_{n-1}}
  \;\cong\;
  \bigoplus_{\substack{\mu\vdash n-1\ \\ \text{s.t.} \, \lambda=\mu+e_j}}
  M_\mu ,
  \label{eq:branching}
\end{equation}
where $j \in \operatorname{add}(\mu)$ runs over all addable rows of $\mu$. In other words, the sum is running over the diagrams $\mu\vdash n-1$ obtained from $\lambda$ by removing one box, each appearing exactly once \cite{fulton1991representation}. Thus the label of all $n$ registers is always the label of the first $n-1$ registers plus one box.

\begin{remark}
    As illustrated in \cref{fig:hybrid}, in \cref{subsec:hybrid} we actually proceed the other way around: we start with a noisy state and introduce symmetry into it by replacements, one register at a time. This is completely analogous, and utilizes the same structure of the branching rule.
\end{remark}

\subsection{Weak Schur sampling and the spectrum estimator}
\label{subsec:wss}

\emph{Weak Schur sampling} is the projective measurement
$\{\Pi_\lambda\}_\lambda$ on $(\C^d)^{\otimes n}$, where $\Pi_\lambda$ is the orthogonal projector
onto the block $V_\lambda\otimes M_\lambda$ of \cref{eq:sw}
\cite{keyl2001estimating,bacon2006efficient}.  It can be implemented by the so-called \textit{Schur transform}, the unitary
that performs the block decomposition of \cref{eq:sw}, using
$\poly(n,d,\log(1/\alpha))$ gates \cite{bacon2006efficient} or, with a
different construction, $\poly(n,\log d,\log(1/\alpha))$ gates
\cite{burchardt2025highdimensionalquantumschurtransforms, Krovi2019}, where $\alpha$ is the accuracy to which the transform is approximated by the circuit.  For an arbitrary input
state $\omega$ on $(\C^d)^{\otimes n}$ the outcome law is
\begin{equation}
  P(\lambda;\omega)\defeq\tr[\Pi_\lambda\,\omega],
  \label{eq:outcomelaw}
\end{equation}
and the \emph{spectrum estimator} is $\hat\pp\defeq\lambda/n$. 

Intuitively, changing the eigenbasis of the target state should not affect
the complexity of learning its spectrum.  Indeed, for a unitary $U\in U(d)$ the operator $U^{\otimes n}$ is of the form \cref{eq:glaction} and hence block diagonal with respect to \cref{eq:sw}, so it commutes with every $\Pi_\lambda$ and
\begin{equation}
  P\bigl(\lambda;\,U^{\otimes n}\omega\,(U^{\otimes n})^\dagger\bigr)
  =P(\lambda;\omega)
  \qquad\text{for all }U\in U(d),
  \label{eq:unitaryinv}
\end{equation}
for every input $\omega$, product or not. We can show an analogous statement for the commutant: permutations do not affect the output distribution on the spectrum, regardless of whether the state is permutationally symmetric on input.

\begin{lemma}[Permutation invariance of weak Schur sampling]
  \label{lem:perminv}
  For every $\pi\in S_n$ and every label $\lambda$,
  \begin{equation}
    \pi\,\Pi_\lambda\,\pi^\dagger=\Pi_\lambda ,
    \qquad\text{equivalently}\qquad
    \pi\,\Pi_\lambda=\Pi_\lambda\,\pi .
    \label{eq:permcommute}
  \end{equation}
  More precisely, with respect to the block decomposition \cref{eq:sw} the
  permutation acts within each block as
  \begin{equation}
    \pi\big|_{V_\lambda\otimes M_\lambda}
    \;=\;\id_{V_\lambda}\otimes\sigma_\lambda(\pi),
    \label{eq:permblock}
  \end{equation}
  where $\sigma_\lambda$ denotes the irreducible representation of $S_n$ on
  $M_\lambda$.  Consequently, for every state $\omega$ on
  $(\C^d)^{\otimes n}$,
  \begin{equation}
    P(\lambda;\pi\,\omega\,\pi^\dagger)=P(\lambda;\omega)
    \qquad\text{for all }\pi\in S_n .
    \label{eq:symmetrisation}
  \end{equation}
\end{lemma}

\begin{proof}
  By \cref{eq:sw}, every permutation is block diagonal, acting on the block
  $V_\lambda\otimes M_\lambda$ within the factor $M_\lambda$ only, which is
  \cref{eq:permblock}.  In particular each block is invariant under $\pi$,
  and since $\pi$ is unitary the orthogonal complement of the block is
  invariant as well.  A unitary that maps both the range and the kernel of
  a projector to themselves commutes with it, which gives
  \cref{eq:permcommute}.  Finally,
  \begin{equation*}
    P(\lambda;\pi\,\omega\,\pi^\dagger)
    =\tr\bigl[\Pi_\lambda\,\pi\,\omega\,\pi^\dagger\bigr]
    =\tr\bigl[\pi^\dagger\,\Pi_\lambda\,\pi\;\omega\bigr]
    =\tr\bigl[\Pi_\lambda\,\omega\bigr]
    =P(\lambda;\omega),
  \end{equation*}
  using cyclicity of the trace and \cref{eq:permcommute}.
\end{proof}

\noindent In words, moving a permutation through $\Pi_\lambda$ only
reshuffles the $M_\lambda$ register inside the block and leaves the block
itself, and hence the measurement outcome, untouched.  So the outcome law of
a product input does not depend on the order of the factors.

The single fact about the noiseless input that we shall need is the
following concentration bound.

\begin{theorem}[Keyl--Werner concentration, {\cite{keyl2001estimating}}, see also {\cite[Thm.~1]{christandl2006spectra}}]
  \label{thm:kw}
  For every state $\varsigma$ on $\C^d$, every number of copies $n \geq 1$ and every precision $\delta > 0$,
  \begin{equation}
    \Pr_{\mathrm{meas}}^{\,\varsigma^{\otimes n}}
    \bigl[\tfrac12\|\lambda/n-\spec(\varsigma)\|_1>\delta\bigr]
    \;\le\; a(n,d)\,\exp(-2n\delta^2),
    \qquad a(n,d)=(n+1)^{O(d^2)},
    \label{eq:precision}
  \end{equation}
  so that the error probability is exponentially suppressed in $n\delta^2$, up to a polynomial prefactor in $n$.
\end{theorem}

Here and throughout, $\Pr^{\,\omega}_{\mathrm{meas}}$ and
$\mathbb{E}^{\,\omega}_{\mathrm{meas}}$ refer to the randomness of the
measurement outcome $\lambda\sim P(\,\cdot\,;\omega)$ for a fixed input
$\omega$.  We package the noiseless performance into a single quantity.

\begin{lemma}[Noiseless expected error]
  \label{lem:perfecterror}
  Define
  \begin{equation}
    \eta(n,d)
    \;\defeq\;
    \sup_{\varsigma}\;
    \mathbb{E}^{\,\varsigma^{\otimes n}}_{\mathrm{meas}}
    \bigl[\tfrac12\|\lambda/n-\spec(\varsigma)\|_1\bigr],
    \label{eq:eta}
  \end{equation}
  the supremum being over all states $\varsigma$ on $\C^d$.  Then
  $\eta(n,d)=O(d/\sqrt n)$.
\end{lemma}

\begin{proof}
  See \cref{app:Keyl_QSE_precision} or \cite{keyl2001estimating}.
\end{proof}

\begin{remark}[The full Keyl measurement]
  \label{rem:fullkeyl}
  Keyl's tomography protocol \cite{keyl2006quantum,odonnell2016efficient}
  performs weak Schur sampling and then a second, covariant measurement on the
  $V_\lambda$ register of the post-measurement state, whose outcome is an
  estimate of the eigenbasis of $\rho$.  Marginalising over that second
  outcome returns the identity on $V_\lambda$, so the law of $\lambda$ under
  the full protocol is exactly \cref{eq:outcomelaw}.  Everything in this
  paper concerns the spectrum estimate $\hat\pp$ and therefore applies
  verbatim to the spectral output of the full protocol; the eigenbasis
  output is the principal open direction, see item~(1) of \cref{sec:open}.
  The second measurement and the explicit formula for \cref{eq:outcomelaw}
  on $\rho^{\otimes n}$ are recalled in \cref{app:keyl}.
\end{remark}

\section{Problem setup and the definition of robustness}
\label{sec:setup}

\subsection{Non-identical input copies: the product model}
\label{subsec:setup}

We consider the setting where the $n$ input copies are \emph{not} identical.
The experimenter intends to prepare $\rho^{\otimes n}$, but the apparatus is
imperfect, and its imperfection may differ from shot to shot.  There are many
ways to model this, and we return to the general picture in
\cref{sec:wasserstein}; to present the ideas in the simplest setting we begin
with the case in which the imperfections of different shots are independent.
We model this by a single \emph{global} error channel $\mathcal{E}$ acting on
all $n$ registers at once, subject to the requirement that it factorises
across registers:
\begin{equation}
  \mathcal{E}=\mathcal{E}_1\otimes\mathcal{E}_2\otimes\cdots\otimes\mathcal{E}_n,
  \qquad
  \mathcal{E}(\rho^{\otimes n})=\bigotimes_{i=1}^n\mathcal{E}_i(\rho)
  =\bigotimes_{i=1}^n\sigma_i,
  \qquad
  \sigma_i\defeq\mathcal{E}_i(\rho).
  \label{eq:productchannel}
\end{equation}
The factors $\mathcal{E}_i$ are \emph{independent of one another} --- this is
what ``product'' means --- but they need \emph{not} be equal, so the copies
are independent but not identically distributed.  A factor is \textit{admissible} if
it moves the target by at most $\eps$ in trace distance, and we write
\begin{equation}
  \mathcal{C}^{\otimes}_\eps
  \;\defeq\;
  \Bigl\{\,\mathcal{E}=\textstyle\bigotimes_{i=1}^n\mathcal{E}_i \;:\;
   \TD\bigl(\mathcal{E}_i(\rho),\rho\bigr)\le\eps\ \text{ for all } i\in[n]\Bigr\}
  \label{eq:channelclass}
\end{equation}
for the class of \textit{admissible} global channels.

\paragraph{\textit{States versus channels.}}
The supremum over $\mathcal{C}^{\otimes}_\eps$ is the same as a supremum over
all product states $\bigotimes_i\sigma_i$ with $\TD(\sigma_i,\rho)\le\eps$:
every such tuple is realised by an admissible channel, namely the product of
the replacement channels $\,\cdot\,\mapsto\sigma_i$, which are completely
positive and trace preserving and satisfy the constraint by construction.  We
use the two languages interchangeably.

\paragraph{\textit{The expectation is over the measurement only.}}
For a \emph{fixed} channel $\mathcal{E}\in\mathcal{C}^{\otimes}_\eps$ the
input $\omega=\mathcal{E}(\rho^{\otimes n})$ is a fixed state, and the only
remaining randomness is that of the measurement outcome
$\lambda\sim P(\,\cdot\,;\omega)$. For any function $F$ of the outcome,
\begin{equation}
  \mathbb{E}^{\,\omega}_{\mathrm{meas}}\bigl[F(\lambda)\bigr]
  \;=\;
  \sum_\lambda F(\lambda)\,P(\lambda;\omega)
  \;=\;
  \tr\Bigl[\Bigl(\sum_\lambda F(\lambda)\Pi_\lambda\Bigr)\omega\Bigr].
  \label{eq:Emeas}
\end{equation}  
Here $\{\Pi_\lambda\}_\lambda$ form a projective measurement, and
$P(\lambda;\omega)=\tr[\Pi_\lambda\omega]$ is the probability of the
outcome $\lambda$. For a real function $F$ of the outcome, the operator
$\sum_\lambda F(\lambda)\Pi_\lambda$ is Hermitian: it is the \emph{observable} whose eigenvalue on the block $\lambda$ is $F(\lambda)$, and $\tr[(\sum_\lambda F(\lambda)\Pi_\lambda)\,\omega]$ is its expectation value in the state $\omega$ (this observable is later given a name, $\mathcal{O}_F$, in \cref{def:singleslot}).

\paragraph{\textit{Correlated noise.}}
Non-product global channels, i.e.\ those in which the errors on different registers are correlated, are not admissible in \cref{eq:channelclass}. However, they are not excluded from the paper: \cref{sec:wasserstein} extends the main result to arbitrary inputs whose distance from $\rho^{\otimes n}$ is small in the quantum Wasserstein sense, and \cref{sec:beyondproduct} exhibits two natural relaxations under which estimation provably fails.

\subsection{Robustness}
\label{subsec:defrobust}
We give the following definition of robustness of quantum spectrum estimation under the product noise model. Note that the definition easily generalizes to
other noise models.

\begin{definition}[Robustness of spectrum estimation]
  \label{def:robust}
  The spectrum estimator $\hat\pp=\lambda/n$ of weak Schur sampling is
  \emph{robust} to $\eps$-perturbations if
  \begin{equation}
    \sup_{\mathcal{E}\in\mathcal{C}^{\otimes}_\eps}\;
    \mathbb{E}^{\,\mathcal{E}(\rho^{\otimes n})}_{\mathrm{meas}}
    \bigl[\tfrac12\|\hat\pp-\pp\|_1\bigr]
    \;\le\; f(\eps)+g(n,d,\eps),
    \label{eq:defrobust}
  \end{equation}
  where $f(\eps)\to0$ as $\eps\to0$ and $g(n,d,\eps)\to0$ as $n\to\infty$
  for each fixed $d$ and $\eps>0$.
\end{definition}

The requirement it imposes is that the error on the output does not grow with the number $n$ of input copies. This is what we need in practice: as discussed in the introduction, an output error growing with $n$ would, at any fixed per-copy noise level, put an upper bound both on the rank of the states whose spectrum can be estimated and on the achievable precision.

We call $f(\eps)$ the \emph{noise floor}: it is the part of the output error that does not vanish as $n\to\infty$ and is caused by the noise alone, below which no number of additional copies can push the error.

The best one can hope for is $f(\eps)=O(\eps)$ and $g(n,d,\eps)=\widetilde{O}(d/\sqrt n)$, where $\widetilde{O}$ omits polylog factors. To see why: for $g(n,d,\eps)=O(d/\sqrt n)$, this term is the error of the noiseless protocol, and it cannot be improved by much. Full state tomography requires $\Theta(d^2/\delta^2)$ copies \cite{haah2017sample,odonnell2016efficient}, and for spectrum
estimation the optimal sample complexity is $O\Bigl(d^2\min\Bigl\{\frac{1}{(\delta\log d)^4},\;\frac{1}{(\delta\log d)^2}\Bigr\}\Bigr)$ \cite{bakshi2026optimalspectrumestimation}, so the only possible improvement over $O(d/\sqrt n)$ is a factor logarithmic in $d$ for the noiseless problem. As for the first term $f(\eps)$, the floor of order $\eps$ cannot be avoided by any protocol under noisy input, which we prove in (\cref{prop:floor}). Notably, \Cref{thm:robust} establishes exactly this behaviour with $f(\eps)=\eps$ and $g(n,d,\eps)=\eta(n,d) = O(d/\sqrt n)$.

\begin{remark}[On the definition]
  \label{rem:defrobust}
  \Cref{def:robust} is deliberately simple.  We introduce it as a first
  step, to fix a notion around which a discussion of robustness can be
  organised, and not as the final word.  Some readers may find it too weak
  for practical purposes: it only asks that $g(n,d,\eps)$ vanish as
  $n\to\infty$, so that a function such as $g=1/(n-1)$ is admissible although
  it gives an infinite bound for a single copy, and nothing is said about the
  rate.  Others may find it too strong: in practice one may be content with
  an output error that does grow with $n$, provided the growth is slow enough
  that the rank of the learnable states and the achievable precision remain
  acceptable.  Alternative notions along both lines are discussed in
  \cref{sec:open}.  For the results of this paper the distinction is
  immaterial, since \cref{thm:robust} gives the strongest form one could ask
  for, namely the noiseless rate plus an $n$-independent floor.
\end{remark}

\begin{remark}[On our results]
  \label{rem:resrobust}
  Importantly, our main results stay relevant outside of the definition of robustness that we introduce: we show that for product input noise with error $\epsilon$ on each register the output noise of weak Schur sampling is increased by only an additive term $\epsilon$. We then generalize this statement to arbitrary input, with the output error depending on the Wasserstein distance of the input from the source. These statements exist outside of any robustness definitions. The main reason why we refer to our definition for robustness in this context is for interpretation purposes: to showcase that weak Schur sampling under reasonable error models gives an output whose error doesn't scale with the input size, even if the input error does scale with the input size. In this sense, it is robust and practical.
\end{remark}

\section{Robustness of spectrum estimation under product noise}
\label{sec:main}

The proof has three ingredients.  A perturbation of a single register enters
the expectation \cref{eq:Emeas} through an effective $d\times d$ observable,
the \emph{single-slot operator} (\cref{subsec:singleslot}); a traceless
perturbation only sees the eigenvalue \emph{spread} of that observable
(\cref{lem:centering}); and the spread is at most $1/n$ for the functions we
care about (\cref{lem:spread}).  Replacing the $n$ perturbed registers by
$\rho$ one at a time then costs at most $\eps$ in total
(\cref{prop:telescope}), and what remains is the noiseless protocol.  The
first and third ingredients are, in a different language, the standard
argument behind the quantum Wasserstein distance of order one
\cite{depalma2021wasserstein}; we spell them out here in the form we need and
make the connection precise in \cref{sec:wasserstein}.

\subsection{Single-slot operators}
\label{subsec:singleslot}

\begin{definition}[Outcome observables and single-slot operators]
  \label{def:singleslot}
  For a real-valued function $f$ on partitions of $n$ with at most $d$ parts, i.e.\ $f\colon\{\lambda\vdash n:\ell(\lambda)\le d\}\to\R$, the
  \emph{outcome observable} is
  \begin{equation}
    \mathcal{O}_f\defeq\sum_\lambda f(\lambda)\,\Pi_\lambda,
    \qquad\text{so that}\qquad
    \mathbb{E}^{\,\omega}_{\mathrm{meas}}[f(\lambda)]=\tr[\mathcal{O}_f\,\omega]
    \label{eq:outcomeobs}
  \end{equation}
  by \cref{eq:Emeas}.  Given a state $R$ on the $n-1$ registers other than
  $i$, the \emph{single-slot operator} is the $d\times d$ Hermitian operator
  \begin{equation}
    h^f_i(R)\defeq\tr_{\neq i}\bigl[\mathcal{O}_f\,(R\otimes_i\id)\bigr],
    \qquad\text{characterised by}\qquad
    \tr\bigl[\sigma\,h^f_i(R)\bigr]=\tr\bigl[\mathcal{O}_f\,(R\otimes_i\sigma)\bigr]
    \ \ \text{for all }\sigma,
    \label{eq:hi}
  \end{equation}
  where $R\otimes_i\sigma$ denotes the state with $\sigma$ in register $i$ and
  $R$ in the others.  It is the effective observable that the collective
  measurement presents to register $i$ once the other registers are fixed.
\end{definition}

\subsection{Only the spread matters}

For a Hermitian operator $h$ with eigenvalues $h_{\max}\ge\cdots\ge h_{\min}$
we write $\operatorname{spread}(h)\defeq h_{\max}-h_{\min}$.

\begin{lemma}[Only the spread of $h$ matters]
  \label{lem:centering}
  Let $A$ be Hermitian and traceless, and let $h$ be Hermitian.  Then
  \begin{equation}
    \bigl|\tr[A\,h]\bigr|\;\le\;\tfrac12\,\|A\|_1\cdot\operatorname{spread}(h).
    \label{eq:centering}
  \end{equation}
\end{lemma}

\begin{proof}
  Since $\tr A=0$ we have $\tr[Ah]=\tr[A(h-c\,\id)]$ for every $c\in\R$.  We
  use H\"older's inequality for the trace pairing \cite[Ch.~IV]{bhatia1997matrix}:
  for operators $A,B$ and $p,q\in[1,\infty]$ with $1/p+1/q=1$,
  \begin{equation}
    \bigl|\tr[A^\dagger B]\bigr|\;\le\;\|A\|_p\,\|B\|_q,
    \label{eq:holder}
  \end{equation}
  where $\|\cdot\|_p$ is the Schatten $p$-norm, i.e.\ the $\ell^p$ norm of the
  singular values; $\|\cdot\|_1$ is the trace norm and
  $\|\cdot\|_\infty=\|\cdot\|_{\mathrm{op}}$ the operator norm.  The case
  $p=1$, $q=\infty$ gives
  $|\tr[A(h-c\,\id)]|\le\|A\|_1\,\|h-c\,\id\|_{\mathrm{op}}$ for every $c$.
  The operator norm of the Hermitian operator $h-c\,\id$ is
  $\max(h_{\max}-c,\;c-h_{\min})$, which is minimised at the midpoint
  $c^\star=(h_{\max}+h_{\min})/2$, where it equals
  $\operatorname{spread}(h)/2$.
\end{proof}

The relevance of \cref{lem:centering} is that the difference of two states is
always traceless, so when we replace one register by another only the spread
of $h^f_i$ enters --- never its size.  This is what makes the following bound
strong enough to be useful: $h^f_i$ is an $O(1)$ operator, but its spread
will turn out to be $O(1/n)$.  The same centring argument is what turns the
quantum Lipschitz constant of \cite{depalma2021wasserstein} into an
operator-norm quantity; see \cref{def:qlipschitz} below.

\subsection{One-box variation and the spread bound}
\label{subsec:spread}

The spread of $h^f_i$ will be controlled by how much $f$ can change when a
single box of the Young diagram is moved.

\begin{definition}[One-box variation]
  \label{def:oneboxvar}
  For $f$ as in \cref{def:singleslot}, set
  \begin{equation}
    \Delta(f)
    \;\defeq\;
    \max_{\substack{\mu\vdash n-1\\ \ell(\mu)\le d}}\;
    \max_{j,j'\in\operatorname{add}(\mu)}\;
    \bigl|f(\mu+e_j)-f(\mu+e_{j'})\bigr|,
    \label{eq:oneboxvar}
  \end{equation}
  where $\operatorname{add}(\mu)\subseteq[d]$ is the set of rows $j$ such
  that $\mu+e_j$ is again a partition (with at most $d$ parts).  Thus
  $\Delta(f)$ is the largest difference of $f$ between two diagrams that share
  the same first $n-1$ boxes.
\end{definition}

\begin{lemma}[One-box variation of the two functions we use]
  \label{lem:onebox}
  \hfill
  \begin{enumerate}
    \item[(i)] For $f(\lambda)=\lambda_k/n$, $\Delta(f)\le1/n$, with equality
      whenever $n\ge k$.
    \item[(ii)] For $f(\lambda)=\tfrac12\|\lambda/n-q\|_1$ with any fixed
      $q\in\R^d$, $\Delta(f)\le1/n$.
  \end{enumerate}
\end{lemma}

\begin{proof}
  Fix $\mu\vdash n-1$ and $j\neq j'$ in $\operatorname{add}(\mu)$.

  (i)~Directly,
  $f(\mu+e_j)-f(\mu+e_{j'})=(\delta_{jk}-\delta_{j'k})/n\in\{0,\pm1/n\}$; the
  value $1/n$ is attained by any $\mu$ with $k\in\operatorname{add}(\mu)$ and
  a second addable row (row $1$ is always addable, and for $k=1$ take
  $\mu=(n-1)$), which exists as soon as $n\ge k$.

  (ii)~By the reverse triangle inequality
  $\bigl|\,\|v-q\|_1-\|w-q\|_1\,\bigr|\le\|(v-q)-(w-q)\|_1=\|v-w\|_1$,
  \begin{equation}
    \bigl|f(\mu+e_j)-f(\mu+e_{j'})\bigr|
    =\frac{1}{2}\left| \, \left\|\frac{\mu+e_j}{n}-q\right\|_1 - \left\|\frac{\mu+e_{j'}}{n}-q\right\|_1 \, \right|
    \le
    \frac{1}{2}\left\|\frac{(\mu+e_j) - (\mu+e_{j'})}{n} \right\|_1
    =\frac{1}{2} \cdot\frac{2}{n}
    =\frac1n .
    \label{eq:step6}
  \end{equation}
\end{proof}

We now prove the key structural lemma.  It rests on the branching rule for the
chain $S_{n-1}\subset S_n$: the label $\lambda$ of all $n$ registers can be
measured jointly with the label $\mu$ of the first $n-1$ registers, and the
two differ by exactly one box.  Since a single register contributes only one
box out of $n$, it can move any function of $\lambda$ only by its one-box
variation.

\begin{lemma}[Spread bound for single-slot operators]
  \label{lem:spread}
  For any function $f$ as in \cref{def:singleslot}, any register $i\in[n]$, and any state
  $R$ on the remaining $n-1$ registers
  \begin{equation}
    \operatorname{spread}\bigl(h^f_i(R)\bigr)\;\le\;\Delta(f).
    \label{eq:spreadbound}
  \end{equation}
  In particular, for $f(\lambda)=\lambda_k/n$ and for
  $f(\lambda)=\tfrac12\|\lambda/n-q\|_1$ with any fixed $q\in\R^d$,
  \begin{equation}
    \operatorname{spread}\bigl(h^f_i(R)\bigr)\;\le\;\frac1n .
    \label{eq:spreadcor}
  \end{equation}
\end{lemma}

\begin{proof}
\hfill

  \step{0}{reduction to the last register} By
  \cref{lem:perminv}, $\mathcal{O}_f$ commutes with the permutation
  operators, so conjugating by the transposition $(i\;n)$ leaves
  $\mathcal{O}_f$ invariant while moving register $i$ to the last position.
  We may therefore assume $i=n$ and write $h\defeq h^f_n(R)$.

  \step{1}{two commuting labels}  Weak Schur sampling on the first
  $n-1$ registers only has projectors $\Pi^{(n-1)}_\mu\otimes\id$,
  $\mu\vdash n-1$.  Each $\Pi^{(n-1)}_\mu$ is a linear combination of
  permutation operators of the first $n-1$ registers (as is every projector
  onto a block of a group action), and $\Pi_\lambda$ commutes with every
  permutation by \cref{lem:perminv}.  Hence the two projectors commute, their product
  \begin{equation}
    Q_{\mu,\lambda}\defeq\bigl(\Pi^{(n-1)}_\mu\otimes\id\bigr)\,\Pi_\lambda
      \label{eq:Qdef}
  \end{equation}
  is again an orthogonal projector, and $\{Q_{\mu,\lambda}\}_{\mu,\lambda}$ is
  a projective measurement refining $\{\Pi_\lambda\}$, with
  $\sum_\lambda Q_{\mu,\lambda}=\Pi^{(n-1)}_\mu\otimes\id$ and $\sum_\mu Q_{\mu,\lambda}=\Pi_\lambda$.  In other words, the label $\lambda$ of all $n$ registers and the label $\mu$ of the first $n-1$ can be measured jointly.

  \step{2}{branching: one box} Recall the branching rule from \cref{subsec:sw}. On the block $V_\lambda\otimes
  M_\lambda$ the subgroup $S_{n-1}$ acts on the factor $M_\lambda$ alone, and the branching rule \cite{fulton1991representation} says that the
  restriction of $M_\lambda$ to $S_{n-1}$ is
  $\bigoplus_{\mu}M_\mu$, the sum running over the diagrams $\mu$ obtained
  from $\lambda$ by removing one box, each with multiplicity one.  Hence
  $Q_{\mu,\lambda}=0$ unless $\lambda=\mu+e_j$ for some
  $j\in\operatorname{add}(\mu)$.

  \step{3}{the other registers fix $\mu$}  On a product input
  $R\otimes\sigma$,
  \begin{equation*}
    P(\mu)\defeq\tr\bigl[(\Pi^{(n-1)}_\mu\otimes\id)(R\otimes\sigma)\bigr]
    =\tr\bigl[\Pi^{(n-1)}_\mu R\bigr]\cdot\tr[\sigma]
    =\tr\bigl[\Pi^{(n-1)}_\mu R\bigr],
  \end{equation*}
  which does not depend on $\sigma$: the label of the first $n-1$ registers is
  determined by those registers alone, and $\{P(\mu)\}_\mu$ is a probability
  distribution.

  \step{4}{the last register only chooses the box}  For
  $j\in\operatorname{add}(\mu)$ define
  \begin{equation*}
    A_{\mu,j}\defeq\tr_{1..n-1}\bigl[Q_{\mu,\mu+e_j}\,(R\otimes\id)\bigr],
    \qquad\text{so that}\qquad
    \tr\bigl[\sigma A_{\mu,j}\bigr]=\tr\bigl[Q_{\mu,\mu+e_j}(R\otimes\sigma)\bigr].
  \end{equation*}
  These operators have three properties.  (a)~\emph{Positivity:} writing
  $R=\sum_a r_a\ket{r_a}\bra{r_a}$ with $r_a\ge0$, for any $\ket\psi\in\C^d$
  we get $\bra\psi A_{\mu,j}\ket\psi
  =\sum_a r_a\bra{r_a\otimes\psi}Q_{\mu,\mu+e_j}\ket{r_a\otimes\psi}\ge0$,
  since $Q_{\mu,\mu+e_j}\succeq0$; hence $A_{\mu,j}\succeq0$.
  (b)~\emph{Row sum rule:} by Steps~1 and~2,
  $\sum_{j\in\operatorname{add}(\mu)}Q_{\mu,\mu+e_j}=\sum_\lambda Q_{\mu,\lambda}
  =\Pi^{(n-1)}_\mu\otimes\id$, so $\sum_{j}A_{\mu,j}=P(\mu)\,\id$; in
  particular $0\preceq A_{\mu,j}\preceq P(\mu)\,\id$.
  (c)~\emph{Total sum:} $\sum_{\mu,j}A_{\mu,j}=\bigl(\sum_\mu P(\mu)\bigr)\id=\id$.

  \step{5}{assembling $h$}  By Steps~1 and~2, for every $\sigma$,
  \begin{equation*}
    \tr[\sigma h]
    =\tr\bigl[\mathcal{O}_f(R\otimes\sigma)\bigr]
    =\sum_\lambda f(\lambda)\sum_\mu\tr\bigl[Q_{\mu,\lambda}(R\otimes\sigma)\bigr]
    =\sum_{\mu}\sum_{j\in\operatorname{add}(\mu)}f(\mu+e_j)\,\tr[\sigma A_{\mu,j}].
  \end{equation*}
  For each $\mu$ choose the reference value
  \begin{equation}
    f_0(\mu)\defeq\tfrac12\bigl(\max_{j}f(\mu+e_j)+\min_{j}f(\mu+e_j)\bigr),
    \label{eq:f0}
  \end{equation}
  the maximum and minimum over $j\in\operatorname{add}(\mu)$, and write
  $f(\mu+e_j)=f_0(\mu)+c_{\mu,j}$; then $|c_{\mu,j}|\le\Delta(f)/2$ by
  \cref{eq:oneboxvar}.  Using (b) and $\tr\sigma=1$,
  \begin{equation*}
    \tr[\sigma h]
    =\underbrace{\sum_\mu f_0(\mu)P(\mu)}_{=:\,c_0(R)}
    +\tr\Bigl[\sigma\sum_{\mu,j}c_{\mu,j}A_{\mu,j}\Bigr].
  \end{equation*}
  Since this holds for all $\sigma$,
  \begin{equation*}
    h=c_0(R)\,\id+B,
    \qquad
    B\defeq\sum_{\mu,j}c_{\mu,j}A_{\mu,j}.
  \end{equation*}
  By (a) and (c), $-\tfrac{\Delta(f)}{2}\id\preceq B\preceq\tfrac{\Delta(f)}{2}\id$,
  so $\operatorname{spread}(h)=\operatorname{spread}(B)\le\Delta(f)$, which
  is \cref{eq:spreadbound}.

  \step{6}{the two special cases}  \Cref{lem:onebox} gives
  $\Delta(f)\le1/n$ for both functions, and \cref{eq:spreadcor} follows from
  \cref{eq:spreadbound}.
\end{proof}

\subsection{Sequential replacement of the registers}
\label{subsec:hybrid}

The next step replaces the perturbed registers by $\rho$ one at a time.  This
telescoping over tensor factors is a standard device; in particular it is how
the quantum Wasserstein distance between two product states is bounded by the
sum of the trace distances of their factors in
\cite[Cor.~1]{depalma2021wasserstein}.  We write it out for the functional at
hand.

\begin{proposition}[Telescoping bound]
  \label{prop:telescope}
  Let $\omega=\bigotimes_{i=1}^n\sigma_i$ be a product state with
  $\TD(\sigma_i,\rho)\le\eps$ for all $i$.  Then for every function $f$ of
  the outcome as in \cref{def:singleslot},
  \begin{equation}
    \Bigl|\,\mathbb{E}^{\,\omega}_{\mathrm{meas}}[f(\lambda)]
      -\mathbb{E}^{\,\rho^{\otimes n}}_{\mathrm{meas}}[f(\lambda)]\,\Bigr|
    \;\le\; n\,\eps\,\Delta(f).
    \label{eq:telescope}
  \end{equation}
  In particular, if $\Delta(f)\le1/n$ --- which holds for $f=\lambda_k/n$ and
  for $f=\tfrac12\|\lambda/n-q\|_1$ by \cref{lem:onebox} --- the two
  expectations differ by at most $\eps$.
\end{proposition}

\begin{proof}
  Interpolate between $\omega$ and $\rho^{\otimes n}$ by replacing one
  register at a time: put
  \begin{equation*}
    \omega^{(k)}\defeq\rho^{\otimes k}\otimes\sigma_{k+1}\otimes\cdots\otimes\sigma_n,
    \qquad k=0,\dots,n,
  \end{equation*}
  so that $\omega^{(0)}=\omega$ and $\omega^{(n)}=\rho^{\otimes n}$.
  Consecutive hybrids differ only in register $k$, so by \cref{eq:hi},
  \begin{equation*}
    \tr\bigl[\mathcal{O}_f\,\omega^{(k-1)}\bigr]-\tr\bigl[\mathcal{O}_f\,\omega^{(k)}\bigr]
    =\tr\bigl[(\sigma_k-\rho)\,h^f_k(R_k)\bigr],
    \qquad
    R_k\defeq\rho^{\otimes(k-1)}\otimes\sigma_{k+1}\otimes\cdots\otimes\sigma_n .
  \end{equation*}
  The operator $\sigma_k-\rho$ is Hermitian and traceless with
  $\|\sigma_k-\rho\|_1=2\,\TD(\sigma_k,\rho)\le2\eps$, and $R_k$ is a state on
  the other $n-1$ registers, which is all that \cref{lem:spread} needs.
  \Cref{lem:centering,lem:spread} therefore give
  \begin{equation*}
    \bigl|\tr[(\sigma_k-\rho)\,h^f_k(R_k)]\bigr|
    \;\le\;\tfrac12\cdot2\eps\cdot\Delta(f)
    \;=\;\eps\,\Delta(f).
  \end{equation*}
  Summing the $n$ telescoping terms and using the triangle inequality gives
  \cref{eq:telescope}.
\end{proof}

\subsection{The main result}
\label{subsec:assembly}

\begin{theorem}[Robustness of spectrum estimation for independent non-identical copies]
  \label{thm:robust}
  For every $n\ge1$, every $\eps\ge0$, and every state $\rho$ on $\C^d$,
  \begin{equation}
    \sup_{\mathcal{E}\in\mathcal{C}^{\otimes}_\eps}\;
      \mathbb{E}^{\,\mathcal{E}(\rho^{\otimes n})}_{\mathrm{meas}}
      \bigl[\tfrac12\|\hat\pp-\pp\|_1\bigr]
    \;\le\;
    \underbrace{\eps}_{\substack{\text{(a) \cref{prop:telescope}}}}
    \;+\;
    \underbrace{\eta(n,d)}_{\substack{\text{(b) statistical error of}\\ \text{the noiseless protocol}}}
    \;=\;\eps+O\!\Bigl(\frac{d}{\sqrt n}\Bigr).
    \label{eq:mainbound}
  \end{equation}
  Hence \cref{def:robust} holds with $f(\eps)=\eps$ and
  $g(n,d,\eps)=\eta(n,d)$: the error does not grow with $n$; it decreases at
  the noiseless rate until it saturates at the floor $\eps$.
\end{theorem}

\begin{proof}
  Fix an admissible channel and write
  $\omega=\mathcal{E}(\rho^{\otimes n})=\bigotimes_i\sigma_i$ with
  $\TD(\sigma_i,\rho)\le\eps$.  Let $L(\lambda)\defeq\tfrac12\|\lambda/n-\pp\|_1$,
  so that $\Delta(L)\le1/n$ by \cref{lem:onebox}(ii).  By
  \cref{prop:telescope} and \cref{lem:perfecterror},
  \begin{equation}
    \mathbb{E}^{\,\omega}_{\mathrm{meas}}[L]
    \;\le\;\mathbb{E}^{\,\rho^{\otimes n}}_{\mathrm{meas}}[L]+n\eps\,\Delta(L)
    \;\le\;\eta(n,d)+\eps .
    \label{eq:assemblychain}
  \end{equation}
  The bound is uniform over admissible channels, so it bounds the supremum.
\end{proof}

\begin{remark}[Reading the two terms]
  \label{rem:twoterms}
  Term~(a) is the only place where the non-identical structure of the input
  enters, and it is the term that the branching argument controls: each of
  the $n$ registers is at most $\eps$ away from $\rho$, and by
  \cref{lem:spread} each can displace the expected error by at most its fair
  share $\eps/n$, so the $n$ displacements accumulate to $\eps$ rather than to
  $n\eps$.  Term~(b) is the ordinary statistical error of collective spectrum
  estimation, untouched by the noise.  Only term~(b) depends on $n$, and it
  carries the noiseless $d/\sqrt n$ scaling, so the sample complexity to reach
  precision $\delta>\eps$ is $n=O\bigl(d^2/(\delta-\eps)^2\bigr)$: the
  collective $d^2$ scaling \cite{keyl2001estimating,odonnell2016efficient}
  is preserved, and the noise costs only an additive floor.
\end{remark}

\begin{remark}[Where each hypothesis is used]
  \label{rem:hypuse}
  The product form \cref{eq:productchannel} is used exactly once, to make the
  telescoping of \cref{prop:telescope} available; closeness is used exactly
  once, to bound $\|\sigma_k-\rho\|_1\le2\eps$ in the same proof.  No
  assumption is made about the rank of $\rho$, about its spectral gaps, or
  about the shape of the outcome distribution.  \Cref{sec:wasserstein} shows
  that the product form can be dispensed with entirely.
\end{remark}

\begin{remark}[Beyond the expectation]
  \label{rem:appendixpointer}
  \Cref{thm:robust} bounds an expected error.  The statements an experimenter
  would additionally want --- what a \emph{single} run guarantees, how many
  repetitions give a prescribed confidence, how heavy the tails of the outcome
  distribution are for a fixed noise realisation, and how much the centre of
  that distribution fluctuates when the noise is itself random --- all follow
  from the same single-register bound, \cref{lem:centering,lem:spread}, and
  are collected in \cref{app:concentration}.  In short: a single run is within
  $3(\eps+\eta(n,d))$ of the truth with probability $2/3$, $O(\log(1/\delta))$
  repetitions with a median-type selection rule give confidence $1-\delta$,
  the variance of the outcome under \emph{any} product input is at most $1/n$,
  and for randomly drawn noise the centre of the outcome distribution
  concentrates within $O(\eps/\sqrt n)$ of its mean.
\end{remark}

\subsection{Optimality of the noise floor}
\label{subsec:floor}

The floor $\eps$ in \cref{eq:mainbound} is not an artefact of the proof.  To show that it cannot be improved by \emph{any} protocol we use Le Cam's
two-point method \cite[Sec.~2.2]{tsybakov2009introduction}, which we state in
the quantum form needed here and again in \cref{sec:beyondproduct}.
The two-point argument establishes the floor $\eps$ alone: it presents
identical inputs in the two cases, so it says nothing about the statistical
term $\eta(n,d)$, which is invisible at the level of a single pair of inputs.
That term is optimal for a separate and simpler reason.  The noiseless input
$\rho^{\otimes n}$ is admissible (the adversary may choose the identity
channel), so any protocol must also pay the minimax error of \emph{noiseless}
spectrum estimation, and that error is $\Omega(d/(\sqrt n\log d))$ by the
lower bound of Wang matched in \cite{bakshi2026optimalspectrumestimation}.
Part~(iv) of \cref{prop:floor} combines the two observations: the worst case
is at least the maximum of the two lower bounds, hence at least half their
sum, so both terms of \cref{eq:mainbound} are necessary, and the bound is
optimal up to constants and a factor logarithmic in $d$ in the second term.

\begin{lemma}[Two-point bound]
  \label{lem:lecam}
  Let $\vartheta_1,\vartheta_2$ be two candidate values of an estimand, at
  distance $s$ in a metric $\mathsf{d}$, and let $\omega_1,\omega_2$ be the
  input states an adversary presents in the two cases.  Then for every
  measurement and every estimator $\hat\vartheta$ of the outcome,
  \begin{equation}
    \max_{i\in\{1,2\}}\;\mathbb{E}^{\,\omega_i}_{\mathrm{meas}}
      \bigl[\mathsf{d}(\hat\vartheta,\vartheta_i)\bigr]
    \;\ge\;\frac{s}{2}\bigl(1-\TD(\omega_1,\omega_2)\bigr).
    \label{eq:lecam}
  \end{equation}
\end{lemma}

\begin{proof}
  Let $P_1,P_2$ be the outcome laws of the measurement on $\omega_1,\omega_2$.
  Pointwise, $\mathsf{d}(\hat\vartheta,\vartheta_1)+\mathsf{d}(\hat\vartheta,\vartheta_2)\ge s$
  by the triangle inequality; integrating against the overlap measure
  $\min(\mathrm{d}P_1,\mathrm{d}P_2)$, whose total mass is
  $1-\mathrm{TV}(P_1,P_2)$, gives
  $\mathbb{E}_{P_1}[\mathsf{d}(\hat\vartheta,\vartheta_1)]
  +\mathbb{E}_{P_2}[\mathsf{d}(\hat\vartheta,\vartheta_2)]\ge s(1-\mathrm{TV}(P_1,P_2))$.
  A measurement is a quantum channel, so
  $\mathrm{TV}(P_1,P_2)\le\TD(\omega_1,\omega_2)$ by data processing
  \cite{nielsen2000}.
\end{proof}

\begin{proposition}[Optimality of both terms in the noisy QSE precision bound]
  \label{prop:floor}
  Let $\sigma$ be any state with $\TD(\sigma,\rho)\le\eps$.
  \begin{enumerate}
    \item[(i)] The identical perturbation $\sigma^{\otimes n}$ is an admissible
      input, $\sigma^{\otimes n}=\mathcal{E}(\rho^{\otimes n})$ for some
      $\mathcal{E}\in\mathcal{C}^{\otimes}_\eps$, and
      \begin{equation}
        \mathbb{E}^{\,\sigma^{\otimes n}}_{\mathrm{meas}}
          \bigl[\tfrac12\|\hat\pp-\pp\|_1\bigr]
        \;\ge\;\tfrac12\|\spec(\sigma)-\pp\|_1-\eta(n,d).
        \label{eq:floorlower}
      \end{equation}
    \item[(ii)] If $\rho$ has full rank and $\eps\le p_d$, the choice
      $\sigma=\rho+\eps\,(\ketbra{v_1}{v_1}-\ketbra{v_d}{v_d})$, where
      $\ket{v_1},\ket{v_d}$ are eigenvectors of $\rho$ for $p_1$ and $p_d$,
      satisfies $\TD(\sigma,\rho)=\eps$ and $\tfrac12\|\spec(\sigma)-\pp\|_1=\eps$.
      Consequently
      \begin{equation}
        \eps-\eta(n,d)
        \;\le\;
        \sup_{\mathcal{E}\in\mathcal{C}^{\otimes}_\eps}\;
        \mathbb{E}^{\,\mathcal{E}(\rho^{\otimes n})}_{\mathrm{meas}}
          \bigl[\tfrac12\|\hat\pp-\pp\|_1\bigr]
        \;\le\;
        \eps+\eta(n,d).
        \label{eq:floortight}
      \end{equation}
    \item[(iii)] The floor is information-theoretic: for $\sigma$ as in (ii),
      every measurement on $n$ copies and every estimator has worst-case
      error at least $\eps/2$ over the two targets $\{\rho,\sigma\}$ with
      admissible adversaries, for every $n$.
    \item[(iv)] Both terms of \cref{eq:mainbound} are necessary.  Let
      $\eta^\star(n,d)$ be the minimax expected error of spectrum estimation
      on the \emph{noiseless} input, i.e.\ the smallest value of
      $\sup_{\varsigma}\mathbb{E}^{\,\varsigma^{\otimes n}}_{\mathrm{meas}}
      [\tfrac12\|\hat\vartheta-\spec(\varsigma)\|_1]$ over all measurements on
      $n$ copies and all estimators $\hat\vartheta$, the supremum running over
      all states $\varsigma$ on $\C^d$.  Then for every measurement and
      estimator, the worst-case error over targets and admissible adversaries
      is at least
      \begin{equation}
        \max\Bigl\{\frac{\eps}{2},\;\eta^\star(n,d)\Bigr\}
        \;\ge\;\frac12\Bigl(\frac{\eps}{2}+\eta^\star(n,d)\Bigr),
        \qquad
        \eta^\star(n,d)=\Omega\Bigl(\frac{1}{\log d}\,
        \min\Bigl\{\frac{d}{\sqrt n},\;\Bigl(\frac{d^2}{n}\Bigr)^{1/4}\Bigr\}\Bigr),
        \label{eq:bothterms}
      \end{equation}
      the last estimate being the inversion of the copy-complexity lower bound
      for noiseless spectrum estimation matched in
      \cite{bakshi2026optimalspectrumestimation}.
      Hence \cref{eq:mainbound} is optimal up to constant factors and a factor
      logarithmic in $d$ in the second term.
  \end{enumerate}
\end{proposition}

\begin{proof}
  (i)~Take every $\mathcal{E}_i$ to be the replacement channel
  $\,\cdot\,\mapsto\sigma$.  Inside the expectation, the reverse triangle
  inequality gives
  $\tfrac12\|\lambda/n-\pp\|_1\ge\tfrac12\|\spec(\sigma)-\pp\|_1
  -\tfrac12\|\lambda/n-\spec(\sigma)\|_1$, and the expectation of the last
  term under $\sigma^{\otimes n}$ is at most $\eta(n,d)$ by
  \cref{lem:perfecterror}.

  (ii)~$\sigma$ is a state: it is diagonal in the eigenbasis of $\rho$ with
  eigenvalues $p_1+\eps\ge p_2\ge\cdots\ge p_{d-1}\ge p_d-\eps\ge0$, summing
  to $1$.  Its trace distance to $\rho$ is
  $\tfrac12\|\eps(\ketbra{v_1}{v_1}-\ketbra{v_d}{v_d})\|_1=\eps$, and since
  the perturbation preserves the ordering of the eigenvalues,
  $\spec(\sigma)-\pp=\eps(e_1-e_d)$, whose $\tfrac12\|\cdot\|_1$ is $\eps$.
  The lower bound in \cref{eq:floortight} is then (i); the upper bound is
  \cref{thm:robust}.

  (iii)~Apply \cref{lem:lecam} with $\vartheta_1=\pp$, $\vartheta_2=\spec(\sigma)$,
  $\mathsf{d}=\tfrac12\|\cdot\|_1$ and $s=\eps$, the adversary presenting
  $\omega_1=\omega_2=\sigma^{\otimes n}$ in both cases --- admissible for the
  target $\rho$ by (i), and trivially admissible (no noise) for the target
  $\sigma$.  The two inputs coincide, so \cref{eq:lecam} gives $\eps/2$.

  (iv)~The identity channel is admissible for every target, so for every
  state $\varsigma$ the noiseless input $\varsigma^{\otimes n}$ is among the
  inputs the adversary may present when the target is $\varsigma$.  The
  worst-case error of any measurement and estimator over targets and
  admissible adversaries is therefore at least its worst-case error on
  noiseless inputs, which is at least $\eta^\star(n,d)$ by definition of the
  minimax error.  Together with (iii) this gives the maximum in
  \cref{eq:bothterms}, and $\max\{a,b\}\ge\tfrac12(a+b)$ for $a,b\ge0$.  For the last estimate, the lower bound of Wang, as stated in
\cite{bakshi2026optimalspectrumestimation}, says that estimating the
spectrum to expected error $\delta$ requires
$\Omega\bigl(d^2\min\{(\delta\log d)^{-4},(\delta\log d)^{-2}\}\bigr)$
copies.  Solving for $\delta$ at $n$ copies gives
$\delta=\Omega\bigl(d/(\sqrt n\log d)\bigr)$ for $n\ge d^2$, where
$\delta\log d\le1$, and $\delta=\Omega\bigl((d^2/n)^{1/4}/\log d\bigr)$ for
$n\le d^2$, where $\delta\log d\ge1$; the two expressions coincide at
$n=d^2$, and in each regime the bound is the smaller of the two, which is
\cref{eq:bothterms}.
\end{proof}

\begin{remark}
  \label{rem:floorreading}
  By Mirsky's inequality \cite[Thm. IV.3.4,Ex. IV.3.5]{bhatia1997matrix} --- eigenvalues
  are $1$-Lipschitz in the trace norm ---
  $\tfrac12\|\spec(\sigma)-\pp\|_1\le\TD(\sigma,\rho)\le\eps$ for every
  admissible $\sigma$, so identical perturbations can never push the error
  above $\eps$, and (ii) shows that they reach $\eps$.  The full-rank condition in (ii) is for convenience of a single formula; for rank-deficient $\rho$ one moves mass $\eps$ from the smallest non-zero eigenvalue into the kernel (for a pure state, $\sigma=(1-\eps)\ketbra{v_1}{v_1}
  +\eps\ketbra{v_2}{v_2}$), which again gives $\tfrac12\|\spec(\sigma)-\pp\|_1=\eps$
  for all sufficiently small $\eps$.  Part~(iii) says that the floor is not
  a weakness of weak Schur sampling: the target is genuinely ambiguous at the
  scale $\eps$, and \cref{thm:robust} shows that weak Schur sampling pays for this exactly once, with the optimal constant.
\end{remark}

\subsection{Randomly drawn channels and unbiased noise}
\label{subsec:random}

Nothing prevents the error channel from itself fluctuating from run to run.
Suppose the factors $\mathcal{E}_i$ in \cref{eq:productchannel} are random,
mutually independent, and admissible almost surely, so that the copies
$\rho_i\defeq\mathcal{E}_i(\rho)$ are mutually independent random states with
$\TD(\rho_i,\rho)\le\eps$.  Write
\begin{equation}
  \rho_i^E\defeq\mathbb{E}_{\mathrm{noise}}[\rho_i]
  \label{eq:rhoE}
\end{equation}
for the noise-averaged copy.  Because $P(\lambda;\cdot)=\tr[\Pi_\lambda\,\cdot\,]$
is linear in the input state, and the $\rho_i$ are independent, the noise
average passes inside and factorises:
\begin{equation}
  \mathbb{E}_{\mathrm{noise}}\bigl[P(\lambda;\rho_1\otimes\cdots\otimes\rho_n)\bigr]
  =\tr\bigl[\Pi_\lambda\,\mathbb{E}_{\mathrm{noise}}[\rho_1\otimes\cdots\otimes\rho_n]\bigr]
  =P\bigl(\lambda;\,\rho_1^E\otimes\cdots\otimes\rho_n^E\bigr).
  \label{eq:meanidentity}
\end{equation}
In words: averaging the outcome law over the noise gives the outcome law of a
single deterministic product input, the product of the noise-averaged copies.

\begin{proposition}[Randomised error channels gain nothing]
  \label{prop:randomnogain}
  For every function $f$ as in \cref{def:singleslot},
  \begin{equation*}
    \sup\;\mathbb{E}_{\mathrm{noise}}\mathbb{E}_{\mathrm{meas}}\bigl[f(\lambda)\bigr]
    \;=\;
    \sup_{\mathcal{E}\in\mathcal{C}^{\otimes}_\eps}\;
      \mathbb{E}^{\,\mathcal{E}(\rho^{\otimes n})}_{\mathrm{meas}}\bigl[f(\lambda)\bigr],
  \end{equation*}
  the left-hand supremum being over all laws of mutually independent random
  channels that are admissible almost surely.  In particular \cref{thm:robust}
  holds for the joint expectation
  $\mathbb{E}_{\mathrm{noise}}\mathbb{E}_{\mathrm{meas}}$ under randomly drawn
  product noise.
\end{proposition}

\begin{proof}
  ($\ge$)~A deterministic channel is a degenerate random one.  ($\le$)~By
  \cref{eq:meanidentity} the joint expectation equals
  $\mathbb{E}^{\,\omega}_{\mathrm{meas}}[F]$ with $\omega=\bigotimes_i\rho_i^E$,
  and $\omega$ is admissible: by convexity of the trace norm (Jensen),
  $\|\rho_i^E-\rho\|_1=\|\mathbb{E}_{\mathrm{noise}}[\rho_i-\rho]\|_1
  \le\mathbb{E}_{\mathrm{noise}}\|\rho_i-\rho\|_1\le2\eps$.
\end{proof}

\begin{definition}[Unbiased noise]
  \label{def:unbiasedbiased}
  A random product error channel is \emph{unbiased} if $\rho_i^E=\rho$ for
  every $i$, i.e.\ the perturbation of each register vanishes on average; it
  is \emph{biased} otherwise.  A deterministic channel is the extreme biased
  case, in which the bias is the entire error.
\end{definition}

\begin{corollary}[Unbiased noise costs nothing on average]
  \label{cor:unbiasedrobust}
  Under unbiased noise the marginal outcome law is \emph{exactly} the
  noiseless one, $\mathbb{E}_{\mathrm{noise}}[P(\lambda;\rho_1\otimes\cdots\otimes\rho_n)]
  =P(\lambda;\rho^{\otimes n})$ for every $\lambda$, and hence
  \begin{equation*}
    \mathbb{E}_{\mathrm{noise}}\mathbb{E}_{\mathrm{meas}}
      \bigl[\tfrac12\|\hat\pp-\pp\|_1\bigr]
    \;=\;
    \mathbb{E}^{\,\rho^{\otimes n}}_{\mathrm{meas}}
      \bigl[\tfrac12\|\lambda/n-\pp\|_1\bigr]
    \;\le\;\eta(n,d),
  \end{equation*}
  an equality with the noiseless error rather than a bound: the floor is $0$.
\end{corollary}

\begin{proof}
  Immediate from \cref{eq:meanidentity} with $\rho_i^E=\rho$.
\end{proof}

\Cref{cor:unbiasedrobust} is a statement about the average over both the noise
and the measurement.  It does not by itself say that the error is small for
one \emph{particular} realisation of the noise; that per-realisation
statement is a concentration result and is proved in
\cref{app:perrealisation}.

\section{Beyond product noise: noisy sources bounded in Wasserstein distance}
\label{sec:wasserstein}

The product model of \cref{sec:setup} was chosen to present the argument in
its simplest form.  Using the language and methods of De~Palma, Marvian,
Trevisan and Lloyd \cite{depalma2021wasserstein}, and the classification of
almost-i.i.d.\ sources of Mazzola, Sutter and Renner \cite{mazzola2026almostiid}
and Girardi, De~Palma, Lami and Datta \cite{girardi2026new,girardi2026robust},
we now generalise \cref{thm:robust} to arbitrary inputs.  The observation
that makes this possible is that our \cref{lem:spread} is really a statement
about the \emph{quantum Lipschitz constant} of the outcome observable
$\mathcal{O}_f$, and that the telescoping of \cref{prop:telescope} is a
special case of the duality seen in \cite{depalma2021wasserstein} between that constant and the quantum Wasserstein
distance of order one.

\subsection{The quantum Wasserstein distance of order one}
\label{subsec:w1}
\hfill
\vspace{2mm}
\\
Before the formal definition, we introduce a bit of intuition for the quantum Wasserstein distance of order one. The Hamming distance between two strings $x,y\in[d]^n$ is the number of positions in which they differ.  For questions about $n$ registers this is often the right
notion of distance, and the trace distance is not: the basis states
$\ket{0000}$ and $\ket{0001}$ are orthogonal, and so are $\ket{0000}$ and
$\ket{1111}$, so the trace distance is $1$ for both pairs, although the first
pair differs in one register and the second in all four.  The quantum
Wasserstein distance of order one of \cite{depalma2021wasserstein} extends
the Hamming distance from strings to arbitrary states of $n$ registers: it
assigns distance $1$ to the first pair and $4$ to the second. It reduces to
the classical Wasserstein (earth mover's) distance with Hamming cost for
states diagonal in the computational basis, and to the trace distance for a single register. More formally, we recall from
\cite{depalma2021wasserstein} exactly the definitions and facts we use.

\begin{definition}[Quantum $W_1$ distance, {\cite[Def.~7]{depalma2021wasserstein}}]
  \label{def:w1}
  For states $\omega,\omega'$ on $(\C^d)^{\otimes n}$,
  \begin{equation}
    \|\omega-\omega'\|_{W_1}
    \;\defeq\;
    \min\Bigl\{\sum_{i=1}^n c_i\;:\;
      c_i\ge0,\ \
      \omega-\omega'=\sum_{i=1}^n c_i\bigl(\tau^{(i)}-\eta^{(i)}\bigr),\ \
      \tau^{(i)},\eta^{(i)}\ \text{states with}\ \tr_i\tau^{(i)}=\tr_i\eta^{(i)}\Bigr\}.
    \label{eq:w1def}
  \end{equation}
\end{definition}

\noindent In words, $\omega-\omega'$ is written as a sum of differences of pairs of
states that agree once register $i$ is traced out, i.e.\ that differ
``only on register $i$'', and the $W_1$ distance is the least total weight
of such a decomposition.  A~perturbation confined to a single register costs
at most its trace distance, however large that is, and the distance adds up
over registers.  The following properties are proved in
\cite{depalma2021wasserstein}.

\begin{proposition}[Basic properties of $W_1$, {\cite[Prop.~2, Prop.~5 and Cor.~1]{depalma2021wasserstein}}]
  \label{prop:w1facts}
  Let $\omega,\omega'$ be states on $(\C^d)^{\otimes n}$.
  \begin{enumerate}
    \item[(i)] $\tfrac12\|\omega-\omega'\|_1\le\|\omega-\omega'\|_{W_1}\le\tfrac n2\|\omega-\omega'\|_1$.
    \item[(ii)] If $\omega=\bigotimes_i\sigma_i$ and $\omega'=\bigotimes_i\rho_i$
      are product states, then
      $\|\omega-\omega'\|_{W_1}=\sum_{i=1}^n\TD(\sigma_i,\rho_i)$.
    \item[(iii)] If $\omega$ and $\omega'$ have the same reduced state on the
      complement of a set $I\subseteq[n]$ of registers, then
      $\|\omega-\omega'\|_{W_1}\le|I|\cdot\|\omega-\omega'\|_1\le2|I|$.
    \item[(iv)] $\|\cdot\|_{W_1}$ is a norm on traceless operators; in
      particular $\omega\mapsto\|\omega-\omega'\|_{W_1}$ is convex.
  \end{enumerate}
\end{proposition}

Property~(ii) is the telescoping of \cref{prop:telescope} in disguise: the
proof in \cite{depalma2021wasserstein} replaces one factor at a time exactly
as we did in our proof.  Property~(iii) says that corrupting $r$ registers
arbitrarily --- even entangling them with each other --- costs $W_1$-distance at
most $2r$, while property~(i) says that a global trace-distance budget
$\theta$ costs at most $n\theta$.

The dual object is a notion of how much an observable can depend on a single register. It provides an operational interpretation of the distance.

\begin{definition}[Quantum Lipschitz constant, {\cite[Def.~8 and Prop.~8]{depalma2021wasserstein}}]
  \label{def:qlipschitz}
  For a Hermitian operator $H$ on $(\C^d)^{\otimes n}$ and a register
  $i\in[n]$, the \emph{dependence of $H$ on register $i$} is
  \begin{equation}
    \partial_iH\;\defeq\;2\,\min_{H'}\bigl\{\|H-H'\|_{\mathrm{op}}\;:\;H'\ \text{acts trivially on register }i\bigr\},
    \label{eq:partial}
  \end{equation}
  and the \emph{quantum Lipschitz constant} of $H$ is
  $\|H\|_L\defeq\max_{i\in[n]}\partial_iH$.  Equivalently
  \cite[Def.~8]{depalma2021wasserstein},
  \begin{equation}
      \|H\|_L=\max_i\max_{\omega,\omega'}\{\tr[H(\omega-\omega')]:\tr_i\omega=\tr_i\omega'\}.
  \end{equation}
\end{definition}

For $n=1$ and $H=h$, the minimum in \cref{eq:partial} is over multiples of
the identity and equals $\operatorname{spread}(h)/2$, so
$\|h\|_L=\operatorname{spread}(h)$: \cref{lem:centering} is the one-register
case of the equivalence of the two expressions in \cref{def:qlipschitz}.
The two notions are dual to each other.

\begin{proposition}[Duality, {\cite[Prop.~9]{depalma2021wasserstein}}]
  \label{prop:duality}
  For states $\omega,\omega'$ on $(\C^d)^{\otimes n}$,
  \begin{equation}
    \|\omega-\omega'\|_{W_1}
    \;=\;\max\bigl\{\tr[H(\omega-\omega')]\;:\;H\ \textnormal{- Hermitian on }(\C^d)^{\otimes n},\ \|H\|_L\le1\bigr\}.
    \label{eq:duality}
  \end{equation}
  Consequently, for every Hermitian $H$,
  \begin{equation}
    \bigl|\tr[H(\omega-\omega')]\bigr|\;\le\;\|H\|_L\,\|\omega-\omega'\|_{W_1}.
    \label{eq:dualitybound}
  \end{equation}
\end{proposition}

\begin{proof}
  \cref{eq:duality} is simply a restatement of {\cite[Prop.~9]{depalma2021wasserstein}}. To show \cref{eq:dualitybound}, we have
  from \cref{eq:partial}, $\|cH\|_L=|c|\,\|H\|_L$ for $c\in\R$, and
  $\|H\|_L=0$ iff $H$ acts trivially on every register, i.e.\ $H=c\,\id$.
  If $\|H\|_L>0$, apply \cref{eq:duality} to $\pm H/\|H\|_L$, which have
  Lipschitz constant $1$; if $\|H\|_L=0$, both sides of
  \cref{eq:dualitybound} vanish since $\tr\omega=\tr\omega'$.
\end{proof}

\subsection{Lipschitz constants of the outcome observables of weak Schur sampling}
\label{subsec:lipschitz}

\hfill
\vspace{2mm}
\\
We now prove the main lemma that enables us to extend the robustness guarantees beyond product noise.

\begin{lemma}[Lipschitz constant of outcome observables]
  \label{lem:qlipschitz}
  For every real-valued function $f$ on partitions of $n$ with at most $d$ parts (\cref{def:singleslot}),
  \begin{equation}
    \|\mathcal{O}_f\|_L\;\le\;\Delta(f).
    \label{eq:qlipschitz}
  \end{equation}
  In particular $\|\mathcal{O}_f\|_L\le1/n$ for $f(\lambda)=\lambda_k/n$ and
  for $f(\lambda)=\tfrac12\|\lambda/n-q\|_1$, $q\in\R^d$.
\end{lemma}

\begin{proof}
  By \cref{lem:perminv}, $\mathcal{O}_f$ is invariant under conjugation
  by permutations, so $\partial_i\mathcal{O}_f$ is the same for all $i$ and
  it suffices to treat $i=n$.  We reuse Steps~1, 2 and~5 of the proof of
  \cref{lem:spread}, this time at the level of operators rather than of
  expectation values.  By Step~1, $\Pi_\lambda=\sum_\mu Q_{\mu,\lambda}$ with
  the mutually orthogonal projectors $Q_{\mu,\lambda}$ of \cref{eq:Qdef}, and
  by Step~2 only $\lambda=\mu+e_j$, $j\in\operatorname{add}(\mu)$, contribute.
  Hence, with $f_0(\mu)$ as in \cref{eq:f0} and $c_{\mu,j}=f(\mu+e_j)-f_0(\mu)$,
  \begin{equation}
    \mathcal{O}_f
    =\sum_{\mu}\sum_{j\in\operatorname{add}(\mu)}f(\mu+e_j)\,Q_{\mu,\mu+e_j}
    =\underbrace{\Bigl(\sum_\mu f_0(\mu)\,\Pi^{(n-1)}_\mu\Bigr)\otimes\id}_{=:\,H'}
    \;+\;
    \underbrace{\sum_{\mu}\sum_{j\in\operatorname{add}(\mu)}c_{\mu,j}\,Q_{\mu,\mu+e_j}}_{=:\,B},
    \label{eq:Ofdecomp}
  \end{equation}
  where we used $\sum_{j}Q_{\mu,\mu+e_j}=\Pi^{(n-1)}_\mu\otimes\id$ from
  Step~4(b).  The operator $H'$ acts trivially on register $n$.  The
  $Q_{\mu,\mu+e_j}$ are mutually orthogonal projectors, so $B$ is Hermitian
  with eigenvalues in $\{c_{\mu,j}\}\cup\{0\}$, and
  $\|B\|_{\mathrm{op}}\le\max_{\mu,j}|c_{\mu,j}|\le\Delta(f)/2$ by
  \cref{eq:oneboxvar}.  Therefore
  $\partial_n\mathcal{O}_f\le2\|\mathcal{O}_f-H'\|_{\mathrm{op}}=2\|B\|_{\mathrm{op}}\le\Delta(f)$.
  The special cases are \cref{lem:onebox}.
\end{proof}

\begin{remark}
  \label{rem:qlipschitzvsspread}
  It's easy to see how
  \Cref{eq:Ofdecomp} also reproves \cref{lem:spread}: for any state
  $R=\sum_ar_a\ketbra{r_a}{r_a}$ on the other registers,
  \begin{equation*}
    h^f_n(R)=\tr\Bigl[\Bigl(\sum_\mu f_0(\mu)\Pi^{(n-1)}_\mu\Bigr)R\Bigr]\,\id
    +\sum_a r_a\,(\bra{r_a}\otimes\id)\,B\,(\ket{r_a}\otimes\id),
  \end{equation*}
  and the second term is a convex combination of compressions of $B$, hence
  has operator norm at most $\|B\|_{\mathrm{op}}\le\Delta(f)/2$ and spread
  at most $\Delta(f)$.  The operator statement
  \cref{eq:qlipschitz} is the stronger one, because it does not require the
  other registers to be in a fixed state.
\end{remark}

\subsection{Robustness for arbitrary inputs}
\label{subsec:w1robust}

\begin{theorem}[Wasserstein robustness of spectrum estimation]
  \label{thm:wasserstein}
  Let $\omega$ be \emph{any} state on $(\C^d)^{\otimes n}$, correlated or
  entangled.  Then for every function $f$ of the outcome as in \cref{def:singleslot},
  \begin{equation}
    \Bigl|\,\mathbb{E}^{\,\omega}_{\mathrm{meas}}[f(\lambda)]
      -\mathbb{E}^{\,\rho^{\otimes n}}_{\mathrm{meas}}[f(\lambda)]\,\Bigr|
    \;\le\;\Delta(f)\,\|\omega-\rho^{\otimes n}\|_{W_1},
    \label{eq:w1telescope}
  \end{equation}
  and in particular
  \begin{equation}
    \mathbb{E}^{\,\omega}_{\mathrm{meas}}\bigl[\tfrac12\|\hat\pp-\pp\|_1\bigr]
    \;\le\;
    \eta(n,d)+\frac1n\,\|\omega-\rho^{\otimes n}\|_{W_1}.
    \label{eq:w1bound}
  \end{equation}
\end{theorem}

\begin{proof}
  \Cref{eq:w1telescope} is \cref{prop:duality} applied to $H=\mathcal{O}_f$,
  together with \cref{lem:qlipschitz}.  For \cref{eq:w1bound} take
  $f=L=\tfrac12\|\lambda/n-\pp\|_1$, for which $\Delta(L)\le1/n$ by
  \cref{lem:onebox}(ii), and bound
  $\mathbb{E}^{\,\rho^{\otimes n}}_{\mathrm{meas}}[L]\le\eta(n,d)$ by
  \cref{lem:perfecterror}.
\end{proof}

The quantity $\eps_{W}(\omega)\defeq\tfrac1n\|\omega-\rho^{\otimes n}\|_{W_1}$
is the natural noise parameter of the general setting: the Wasserstein
distance from the ideal input, per register.  \Cref{thm:wasserstein} says that
weak Schur sampling is robust in the sense of \cref{def:robust} with the
class $\mathcal{C}^\otimes_\eps$ replaced by
$\{\omega:\eps_W(\omega)\le\eps\}$, again with $f(\eps)=\eps$ and
$g(n,d,\eps)=\eta(n,d)$.  Some instances:

\begin{corollary}[Special cases]
  \label{cor:w1cases}
  Let $L=\tfrac12\|\hat\pp-\pp\|_1$.
  \begin{enumerate}
    \item[(i)] \textup{(Product noise with register-dependent budgets.)}  If
      $\omega=\bigotimes_i\sigma_i$ with $\TD(\sigma_i,\rho)\le\eps_i$, then
      $\mathbb{E}^{\,\omega}_{\mathrm{meas}}[L]\le\eta(n,d)+\tfrac1n\sum_i\eps_i$.
      For $\eps_i\equiv\eps$ this is \cref{thm:robust}.
    \item[(ii)] \textup{(A few corrupted registers.)}  If $\omega$ agrees with
      $\rho^{\otimes n}$ on $n-r$ registers, and on the remaining $r$ registers the state is arbitrary (and possibly entangled), then
      $\mathbb{E}^{\,\omega}_{\mathrm{meas}}[L]\le\eta(n,d)+2r/n$; if the
      corrupted registers are moreover in product form with each other and
      with the rest, the bound improves to $\eta(n,d)+r/n$.
    \item[(iii)] \textup{(Global trace-distance budget.)}  If
      $\TD(\omega,\rho^{\otimes n})\le\theta$, then
      $\mathbb{E}^{\,\omega}_{\mathrm{meas}}[L]\le\eta(n,d)+\theta$.
    \item[(iv)] \textup{(Mixtures.)}  If $\omega$ is a convex combination of states each within $W_1$-distance $n\eps$ of $\rho^{\otimes n}$, then $\mathbb{E}^{\,\omega}_{\mathrm{meas}}[L]\le\eta(n,d)+\eps$.
    \item[(v)] \textup{(Wasserstein almost-i.i.d.\ sources.)}  If $(\omega_n)_n$
      is a Wasserstein almost-i.i.d.\ source along $\rho$ in the sense of
      \cite[Def.~7]{girardi2026robust}, i.e.\
      $\tfrac1n\|\omega_n-\rho^{\otimes n}\|_{W_1}\to0$, then
      $\mathbb{E}^{\,\omega_n}_{\mathrm{meas}}[L]\to0$, at rate
      $\max\{\eta(n,d),\,\eps_W(\omega_n)\}$ up to a factor $2$.
  \end{enumerate}
\end{corollary}

\begin{proof}
  (i)~is \cref{prop:w1facts}(ii); (ii)~is \cref{prop:w1facts}(iii), and (ii)
  again in the product case; (iii)~is \cref{prop:w1facts}(i) (it also
  follows directly from data processing, cf.\ \cref{prop:globalbudget}(i));
  (iv)~is \cref{prop:w1facts}(iv); (v)~is the definition.  In each case
  insert the bound on $\|\omega-\rho^{\otimes n}\|_{W_1}$ into
  \cref{eq:w1bound}.
\end{proof}

\begin{remark}[The floor is still optimal]
  \label{rem:w1floor}
  The identical perturbation $\sigma^{\otimes n}$ of \cref{prop:floor} has
  $\|\sigma^{\otimes n}-\rho^{\otimes n}\|_{W_1}=n\,\TD(\sigma,\rho)$ by
  \cref{prop:w1facts}(ii), so $\eps_W=\TD(\sigma,\rho)$ and
  \cref{prop:floor} shows that the floor $\eps_W$ in \cref{eq:w1bound} is
  attained, and that no measurement and estimator can beat $\eps_W/2$ on
  the corresponding two-point instance.
\end{remark}

\begin{remark}[Which perturbations are small in $W_1$]
  \label{rem:w1scale}
  The $W_1$ distance measures a perturbation by the number of registers it
  touches, weighted by how much it moves each: many small local errors, a few
  large local errors, and correlations among the errors of a few registers
  are all cheap, whereas a perturbation that shifts \emph{every} register by a
  constant is expensive.  The marginal counterexample $\tau$ of
  \cref{prop:marginalfailure} below is of the latter kind, even though its
  single-copy marginals are exactly $\rho$: applying \cref{eq:w1telescope}
  with $f=\lambda_1/n$ gives
  \begin{equation*}
    \|\tau-\rho^{\otimes n}\|_{W_1}
    \;\ge\;n\Bigl(\mathbb{E}^{\,\tau}_{\mathrm{meas}}[\lambda_1/n]
      -\mathbb{E}^{\,\rho^{\otimes n}}_{\mathrm{meas}}[\lambda_1/n]\Bigr)
    \;\ge\;n\bigl(1-p_1-2\eta(n,d)\bigr),
  \end{equation*}
  so $\eps_W(\tau)\to1-p_1$: $\tau$ is not Wasserstein almost-i.i.d.\ along
  $\rho$, although it is weakly almost-i.i.d.\ at $k=1$.  This is an instance
  of the strict separation between the two notions proved in
  \cite{girardi2026new}, and it is why \cref{thm:wasserstein} and
  \cref{prop:marginalfailure} are consistent.
\end{remark}

\begin{remark}[Scope of the method]
  \label{rem:method}
  Nothing in \cref{thm:wasserstein} is specific to the spectrum beyond
  \cref{lem:qlipschitz}: for \emph{any} collective protocol whose
  post-processing is a function $f$ of a measurement with outcome observable
  $\mathcal{O}_f$, the error under an input $\omega$ exceeds the noiseless
  error by at most $\|\mathcal{O}_f\|_L\,\|\omega-\rho^{\otimes n}\|_{W_1}$.
  For statistics that are sums of few-body terms --- such as the swap
  statistics used for hypothesis testing in \cite{depalma2025noniid} --- the
  Lipschitz constant is $O(1/n)$ for elementary reasons; the content of
  \cref{lem:qlipschitz} is that the same holds for the genuinely
  $n$-body observables $\lambda_k/n$ and $\tfrac12\|\lambda/n-q\|_1$.
\end{remark}

\section{When is the noise too strong? Marginal and global constraints}
\label{sec:beyondproduct}

\Cref{thm:robust,thm:wasserstein} bound the error by a ``per-register'' budget:
$\eps$ in the product model and $\tfrac1n\|\omega-\rho^{\otimes n}\|_{W_1}$
in general. By ``per-register'' we mean a budget promise that accounts for the noise register by register: in the product model each copy may move by at most $\eps$, and in general the $W_1$ distance $\|\omega-\rho^{\otimes n}\|_{W_1}$ allows correlated and even entangled noise as long as its total weight, spread over the $n$ registers, is $\eps$ per register, that is $\|\omega-\rho^{\otimes n}\|_{W_1}\le n\eps$.

It is natural to ask whether still weaker promises on the input
suffice, i.e.\ if we can prove robustness under stronger noise.  We consider the two most natural ones: constraining only the
single-copy \emph{marginals} of an arbitrary joint state, which is the $k=1$ case of the weakly almost-i.i.d.\ sources of \cite{girardi2026new}, or only the \emph{global} trace distance to the perfect input. In each class we explicit adversarial inputs, with matching lower bounds valid for every measurement and estimator, under which estimation fails.  A single construction witnesses both failures. These examples say that the noise cannot be allowed to be arbitrarily bad, but they do not say that per-register budgets proven robust in this work are the only workable promises, and we return to the space of intermediate models in \cref{sec:open}.

\begin{definition}[Relaxed noise classes]
  \label{def:relaxedclasses}
  Fix the target $\rho$ and $n$.  For $\eps\ge0$ let $\mathcal{M}_\eps$ be
  the class of (arbitrary, possibly entangled and correlated) joint states
  $\omega$ on $(\C^d)^{\otimes n}$ whose every single-copy marginal satisfies
  $\TD(\tr_{\neq i}[\omega],\rho)\le\eps$ for all $i\in[n]$.  For
  $\theta\in[0,\tfrac12)$ let $\mathcal{G}_\theta$ be the class of joint
  states with $\TD(\omega,\rho^{\otimes n})\le\theta$.  The product model
  embeds into $\mathcal{G}_{n\eps}$: replacing one factor at a time,
  \begin{equation}
    \TD\Bigl(\bigotimes_i\sigma_i,\;\rho^{\otimes n}\Bigr)
    \;\le\;\sum_{i=1}^n\TD(\sigma_i,\rho)\;\le\;n\eps .
    \label{eq:hybridbudget}
  \end{equation}
\end{definition}

\begin{proposition}[Marginal constraints admit exact failure at $\eps=0$]
  \label{prop:marginalfailure}
  Let $\rho=\sum_kp_k\ketbra{v_k}{v_k}$ be mixed ($p_1<1$) and define the
  classically correlated eigenbasis mixture
  \begin{equation}
    \tau\;\defeq\;\sum_{k=1}^dp_k\bigl(\ketbra{v_k}{v_k}\bigr)^{\otimes n}.
    \label{eq:tau}
  \end{equation}
  Then:
  \begin{enumerate}
    \item[(i)] $\tau\in\mathcal{M}_0$: every single-copy marginal of $\tau$
      equals $\rho$ exactly.
    \item[(ii)] Weak Schur sampling on $\tau$ outputs $\lambda=(n,0,\dots,0)$
      with probability one, so $\hat\pp=(1,0,\dots,0)$ deterministically and
      $\tfrac12\|\hat\pp-\pp\|_1=1-p_1$ for every $n\ge1$.
    \item[(iii)] The failure is not specific to weak Schur sampling: for every
      measurement and estimator, the worst case over the two targets
      $\{\rho,\ \ketbra{v_1}{v_1}\}$ with adversaries constrained to
      $\mathcal{M}_0$ is at least $\tfrac12p_1(1-p_1)$, independent of $n$.
  \end{enumerate}
\end{proposition}

\begin{proof}
  (i)~$\tr_{\neq i}[\tau]=\sum_kp_k\ketbra{v_k}{v_k}=\rho$ for each $i$.

  (ii)~Each branch $\ket{v_k}^{\otimes n}$ is a pure product state and hence
  lies in the symmetric subspace $\mathrm{Sym}^n(\C^d)=V_{(n)}\otimes M_{(n)}$,
  which is exactly the $\lambda=(n,0,\dots,0)$ block (with $m_{(n)}=1$).
  Thus $\tr[\Pi_{(n)}\tau]=1$ and the outcome is deterministic.  The error is
  $\tfrac12\bigl[(1-p_1)+\sum_{k\ge2}p_k\bigr]=1-p_1$.

  (iii)~Apply \cref{lem:lecam}.  For the target $\vartheta_1=\rho$ the
  adversary plays $\tau$; for the target $\vartheta_2=\ketbra{v_1}{v_1}$ the
  marginal constraint at $\eps=0$ \emph{forces}
  $\omega_2=\ketbra{v_1}{v_1}^{\otimes n}$, since a joint state with pure
  marginals is the product of those marginals.  The two inputs satisfy
  $\TD(\tau,\ketbra{v_1}{v_1}^{\otimes n})=1-p_1$ (the $k\ge2$ branches are
  orthogonal to $\omega_2$, the $k=1$ branch identical), and the target
  spectra are at distance $s=\tfrac12\|\pp-(1,0,\dots,0)\|_1=1-p_1$.  Hence
  the worst-case error is at least $\tfrac s2\bigl(1-(1-p_1)\bigr)=\tfrac12p_1(1-p_1)$.
\end{proof}

\begin{remark}[de Finetti reading; the Haar variant]
  \label{rem:definettireading}
  The state \cref{eq:tau} is exchangeable:
  $\tau=\int\varsigma^{\otimes n}\,\mathrm{d}\nu(\varsigma)$ with $\nu$
  supported on the pure eigenstates of $\rho$.  By linearity of
  $P(\lambda;\cdot)$, the outcome law on any exchangeable input is the
  $\nu$-mixture of noiseless laws,
  $P(\lambda;\tau)=\int P(\lambda;\varsigma^{\otimes n})\,\mathrm{d}\nu(\varsigma)$:
  the measurement estimates the spectrum of the de~Finetti \emph{mixing
  state} $\varsigma\sim\nu$, never the spectrum of the marginal
  $\int\varsigma\,\mathrm{d}\nu$ \cite{christandl2007}.  Marginal closeness is
  therefore the wrong promise for a collective measurement, and
  \cref{prop:marginalfailure}(iii) shows that no measurement repairs this.  A
  continuous variant: for $\rho=\id/d$, the Haar mixture
  $\int\ketbra{\psi}{\psi}^{\otimes n}\,\mathrm{d}\psi
  =\Pi_{\mathrm{sym}}/\dim\mathrm{Sym}^n(\C^d)$ has all marginals exactly
  maximally mixed yet yields $\hat\pp=(1,0,\dots,0)$ deterministically, with
  error $1-1/d$.
\end{remark}

\begin{proposition}[Global budget: tight $\Theta(\theta)$ floor, and growing
  error for $\theta=n\eps$]
  \label{prop:globalbudget}
  Fix $\theta\in(0,\tfrac12)$ and let $\rho$ be mixed.  Over the class
  $\mathcal{G}_\theta$:
  \begin{enumerate}
    \item[(i)] \textup{(Upper bound, weak Schur sampling.)}
      $\displaystyle\sup_{\omega\in\mathcal{G}_\theta}
      \mathbb{E}^{\,\omega}_{\mathrm{meas}}\bigl[\tfrac12\|\hat\pp-\pp\|_1\bigr]
      \le\theta+\eta(n,d).$
    \item[(ii)] \textup{(Lower bound, constructive.)} The state
      $\omega=(1-\theta)\rho^{\otimes n}+\theta\,\tau$, with $\tau$ from
      \cref{eq:tau}, lies in $\mathcal{G}_\theta\cap\mathcal{M}_0$
      simultaneously, and the error of weak Schur sampling on it is at least
      $\theta(1-p_1)$, for every $n$.
    \item[(iii)] \textup{(Lower bound, any measurement and estimator.)} For
      every measurement and estimator, the worst case over the targets
      $\{\rho,\ \ketbra{v_1}{v_1}\}$ with adversaries in the respective
      $\mathcal{G}_\theta$ classes is at least $\theta(1-p_1)$, independent
      of $n$.
  \end{enumerate}
  Consequently, since we expect for the global error $\theta$ to grow with system size, the output error would also grow with system size. For example, with the summed per-copy budget $\theta=n\eps$ of
  \cref{eq:hybridbudget} the worst-case error over $\mathcal{G}_{n\eps}$
  grows linearly,
  \begin{equation}
    \sup_{\omega\in\mathcal{G}_{n\eps}}
      \mathbb{E}^{\,\omega}_{\mathrm{meas}}\bigl[\tfrac12\|\hat\pp-\pp\|_1\bigr]
    \;\asymp\;\min(n\eps,1)
    \qquad\text{(up to $\rho$-dependent constants)},
    \label{eq:lineargrowth}
  \end{equation}
  increasing with $n$ until saturation at $n\approx1/\eps$. Within the global
  model, additional copies strictly \emph{worsen} the guarantee.
\end{proposition}

\begin{proof}
  (i)~The measurement is a channel, so by data processing the outcome laws of
  $\omega$ and $\rho^{\otimes n}$ differ by at most $\theta$ in total
  variation \cite{nielsen2000}; since the loss is $[0,1]$-valued, the
  expectations differ by at most $\theta$.  Then apply \cref{lem:perfecterror}.
  (Alternatively, \cref{cor:w1cases}(iii).)

  (ii)~Trace distance is jointly convex, so
  $\TD(\omega,\rho^{\otimes n})\le\theta\,\TD(\tau,\rho^{\otimes n})\le\theta$;
  the marginals are $(1-\theta)\rho+\theta\rho=\rho$ by
  \cref{prop:marginalfailure}(i).  The error is the mixture
  $(1-\theta)\mathbb{E}^{\,\rho^{\otimes n}}_{\mathrm{meas}}[L]+\theta(1-p_1)\ge\theta(1-p_1)$,
  using \cref{prop:marginalfailure}(ii) for the second branch.

  (iii)~Cross-mixing.  For the target $\vartheta_1=\rho$ the adversary plays
  $\omega_1=(1-\theta)\rho^{\otimes n}+\theta\ketbra{v_1}{v_1}^{\otimes n}$;
  for the target $\vartheta_2=\ketbra{v_1}{v_1}$ she plays $\omega_2$ with the
  two weights exchanged.  Each is admissible for its \emph{own} target by
  joint convexity, $\TD(\omega_i,\vartheta_i^{\otimes n})\le\theta$, yet the
  two are close to \emph{each other}:
  $\omega_1-\omega_2=(1-2\theta)\bigl(\rho^{\otimes n}-\ketbra{v_1}{v_1}^{\otimes n}\bigr)$,
  so $\TD(\omega_1,\omega_2)\le1-2\theta$.  The target spectra are at distance
  $s=1-p_1$, so \cref{lem:lecam} gives worst-case error at least
  $\tfrac s2\bigl(1-(1-2\theta)\bigr)=\theta(1-p_1)$.

  For \cref{eq:lineargrowth}: the upper bound is (i) with
  $\theta=\min(n\eps,1)$ (for $n\eps\ge1$ the constraint is vacuous and the
  error is trivially at most $1$); the lower bound is (iii) with $\theta=n\eps$
  for $n\le1/(2\eps)$, and monotonicity of the classes
  ($\mathcal{G}_\theta\subseteq\mathcal{G}_{\theta'}$ for $\theta\le\theta'$)
  beyond.
\end{proof}

\begin{remark}[What localisation of the budget buys]
  \label{rem:productworth}
  \Cref{prop:globalbudget} makes the role of the per-register structure
  quantitative.  The product-$\eps$ model and the global-$n\eps$ model spend
  the \emph{same} total trace-distance budget \cref{eq:hybridbudget}, yet the
  worst-case errors differ by a factor of $n$: $\Theta(\min(n\eps,1))$ for the
  global class versus $\eps+O(d/\sqrt n)$ for the product class
  (\cref{thm:robust}).  The mechanism is exactly the one isolated by
  \cref{lem:spread,lem:qlipschitz}: when the budget is \emph{localised}, one
  $\eps$ per register, a permutation-invariant measurement lets each register
  move the expected error only by its share $\eps/n$, so the $n$ shares sum to
  $\eps$ --- which is what the $W_1$ distance measures --- whereas a global
  budget may be spent coherently on a single $n$-body direction to which the
  collective measurement is maximally sensitive: in \cref{eq:tau}, classical
  correlation of the eigenbasis choice across all copies.  In particular the
  na\"ive hope that a global budget $\theta$ should cost only $O(\theta/n)$ in
  precision is refuted information-theoretically by
  \cref{prop:globalbudget}(iii).
\end{remark}

\section{Discussion and open questions}
\label{sec:open}

\Cref{thm:robust,thm:wasserstein} settle robustness of the \emph{spectrum}
estimate for product noise and, more generally, for inputs at small
Wasserstein distance from $\rho^{\otimes n}$, unconditionally and with the
optimal floor; \cref{sec:beyondproduct} shows that two natural weaker
promises do not suffice.  We first comment on what these results do and do
not say about the noise model and about the notion of robustness itself, and
then list open questions.

\subsection*{Intermediate noise models}

The negative results of \cref{sec:beyondproduct} are not converses to the
positive ones. They show that the promise on the input cannot be too weak: closeness of single-copy marginals is compatible with a constant error at $\eps=0$, and a global trace-distance budget is compatible with an error that grows with $n$. Nevertheless, there is a large space of models between these extremes and the per-register budgets of \cref{thm:robust,thm:wasserstein},
and some of them are physically natural.  Constraining the average $k$-body
marginals to be close to $\rho^{\otimes k}$ for a fixed $k\ge2$ (the weakly
almost-i.i.d.\ sources of \cite{girardi2026new} at finite $k$) already
excludes the counterexample \cref{eq:tau}, whose two-body marginal carries
$\Theta(1)$ classical correlation; the almost-i.i.d.\ states of Mazzola,
Sutter and Renner \cite{mazzola2026almostiid}, in which all but $r=o(n)$
systems are in the ideal state in a symmetrised sense, form another class,
which for mixtures of product corruptions is covered by
\cref{cor:w1cases}(ii) but in general is not obviously so; and $k$-wise
independent error channels interpolate between the product model and
arbitrary correlations.  For which of these classes weak Schur sampling
remains robust, and with what noise floor, is open.  \Cref{rem:w1scale} gives the
mechanism to look for: the question is always whether the relevant class stays
within $W_1$-distance $O(n\eps)$ of $\rho^{\otimes n}$, or, if it does not,
whether a different Lipschitz-type structure of the outcome observables can
be exploited.

\subsection*{Alternative notions of robustness}

\Cref{def:robust} demands that for fixed $\eps$ the output error does not grow with the input size $n$, even if the input error does grow with $n$. This is the natural requirement when copies are cheap and one wants the
collective advantage to survive at arbitrary $n$, but it is not the only
reasonable one, and in the biased noise setting it is in some respects too
demanding.  A first observation is that the noise floor $\eps$ and the target precision $\delta$ interact: whatever the noise model, an experimenter who wants precision $\delta$ never needs more than $n\approx d^2/\max(\delta,\eps)^2$ copies, because beyond that point the noiseless error has already dropped to the level of the floor and, in the biased model, no number of additional copies can improve on $\max(\delta,\eps)$ (\cref{prop:floor}(iii)). A definition of robustness with such a cap on $n$ could admit models in which the worst-case error
does grow with $n$, as in \cref{prop:globalbudget}, provided the growth is
slow enough to be harmless up to the desired precision.  Similarly, one may accept floors that depend on the local dimension $d$, or on spectral gaps of $\rho$, or on other parameters that admit physical constraints in practice. Which of these relaxations are meaningful for applications, and which noise classes they admit, seems to us a worthwhile question in its own right.

\subsection*{Open questions}
\begin{enumerate}
\item \textbf{Robustness of the eigenbasis estimate: full state tomography.}
  This is one of the two main open directions and the original motivation
  for this work.  The full Keyl measurement (\cref{app:keyl}) returns the spectrum and basis pair $(\hat\pp,\hat U)$, and the present analysis covers only $\hat\pp$.
  Every quantity bounded in \cref{sec:main,sec:wasserstein} is a linear
  functional $\tr[\mathcal{O}_F\,\omega]$ of the input, and both the
  telescoping of \cref{prop:telescope} and the duality of
  \cref{prop:duality} rely on that linearity; the eigenbasis estimate is a
  nonlinear function of the post-measurement state inside the block
  $\lambda$, so the same argument does not apply to it.  What remains to be
  established is therefore the robustness of the second measurement stage,
  the eigenbasis measurement, on a post-measurement state that is no longer
  of the ideal form. The random purification route of \cite{pelecanos2025mixed}
  offers a second path: there the only collective step is weak Schur
  sampling followed by a coherent purification, and the pure-state
  tomography that follows is comparatively simple.

\item \textbf{Robustness of the optimal spectrum estimation algorithm.}
  The sample-optimal algorithm of Bakshi, Singh and Tan
  \cite{bakshi2026optimalspectrumestimation} uses weak Schur sampling as its
  core subroutine but adds two ingredients: a preliminary bucketing stage,
  which learns the subspace of large eigenvalues by tomography and filters
  the copies accordingly \cite{pelecanos2026keylwerner}, and a
  moment-matching post-processing of the outcome $\lambda$, which recovers
  the small eigenvalues from Chebyshev moments of the spectrum.  Our results
  cover the weak Schur sampling stage.  What remains to be established is
  the robustness of the bucketing stage and of the post-processing, that is,
  that neither amplifies the per-copy noise.  If both hold, the robustness
  claim of this paper extends to the optimal algorithm, with the same noise floor and the optimal noiseless rate in place of $\eta(n,d)$.

\item \textbf{Intermediate noise classes.}  As discussed above: determine
  for which of the classes between per-register budgets and the failing
  classes of \cref{sec:beyondproduct} --- $k$-body marginal constraints,
  Mazzola--Sutter--Renner almost-i.i.d.\ states, $k$-wise independent
  channels, sublinear $W_1$ budgets --- weak Schur sampling is robust, and
  how the error interpolates between $\Theta(\eps)$ and
  $\Theta(\min(n\eps,1))$.

\item \textbf{The exact worst case at finite $n$.}  \Cref{prop:floor} pins
  the supremum in \cref{thm:robust} to the interval
  $[\eps-\eta(n,d),\,\eps+\eta(n,d)]$, so the floor is settled.  Two finer
  questions remain.  First, is the noisy error ever \emph{larger} than
  $\max(\eps,\eta(n,d))$ up to constants, i.e.\ do the two terms in
  \cref{eq:mainbound} genuinely add, or does the true worst case behave like
  their maximum?  The identical perturbation of \cref{prop:floor} gives
  $\tfrac12\|\spec(\sigma)-\pp\|_1\pm\eta(n,d)$, which does not decide this.
  Second, the telescoping bounds each of the $n$ replacement terms by its
  absolute value and thereby discards any cancellation between registers
  perturbed in opposite directions; whether an adversary can prevent
  cancellation while keeping every register within $\eps$, other than by
  perturbing all registers identically, is open.  Relatedly,
  \cref{lem:qlipschitz} gives $\|\mathcal{O}_{\lambda_k/n}\|_L\le1/n$;
  computing the exact Lipschitz constant would show how much room there is.

\item \textbf{Finer distributional information for the noiseless
  measurement.}  The proof of \cref{thm:robust} needs only the two-moment
  fact $\eta(n,d)=O(d/\sqrt n)$.  Per-realisation statements about the
  \emph{shape} of the outcome distribution --- for example a local central
  limit theorem for the law \cref{eq:keylformula} at fixed $d$, i.e.\ a
  pointwise Gaussian approximation with multiplicative error $1+O(n^{-1/2})$
  on the region $\|\lambda/n-\pp\|_\infty=O(n^{-1/2})$ --- would sharpen the
  tail bounds of \cref{app:concentration} and are, to our knowledge, not in
  the literature.  The available central limit theorems for Schur--Weyl
  measures concern the regime in which $d$ grows with $n$ (typically
  $d\propto\sqrt n$), where the limit shape is a deformed Vershik--Kerov
  curve; the regime here is $d$ and $\pp$ fixed with $n\to\infty$, in which
  $\lambda/n\to\pp$ and the fluctuation theory is multinomial-like.  The
  moment computations of \cite{odonnell2016efficient,haah2017sample} are the
  natural starting point.

\item \textbf{Other collective protocols.}  By \cref{rem:method} the method
  applies to any collective protocol whose classical output is a function of
  a measurement outcome, once the Lipschitz constant of the outcome
  observable is known.  Natural targets are the entropy and purity estimators
  built on weak Schur sampling \cite{acharya2020measuring,odonnell2015quantum},
  the purification step of \cite{pelecanos2025mixed}, and the mixed Schur
  sampling of \cite{cervero2024memory} for the walled Brauer setting of
  process tomography \cite{koike1989}, where the one-box variation
  \cref{eq:oneboxvar} would have to be recomputed for the corresponding
  diagram combinatorics.

\item \textbf{A classical analogue.}  When all $\sigma_i$ commute with $\rho$,
  the outcome law of weak Schur sampling on $\bigotimes_i\sigma_i$ is the law
  of the RSK shape of a word of $n$ independent, non-identically distributed
  letters.  \Cref{thm:robust} then says that the RSK shape of a word whose
  letter distributions are each $\eps$-close to $\pp$ is, in expectation,
  within $\eps+O(d/\sqrt n)$ of the shape of an i.i.d.\ word from $\pp$ ---
  a statement about the Robinson--Schensted--Knuth correspondence under
  non-stationary input that may be of independent combinatorial interest.
\end{enumerate}

\section*{Acknowledgements}
\noindent VV is grateful to Filippo Girardi for the insightful discussions, for pointing us towards the concept of the quantum Wasserstein distance, and for suggesting some direct applications. We thank Elias Theil, Daniel Stilck Fran\c{c}a, and Yuming Zhao for all the helpful discussions. VV also thanks Xinyu (Norah) Tan for suggesting the direct application of our results towards the robustness question of their random dimension reduction problem.
\\
\\
\noindent \textbf{AI disclosure.} Introducing the question of robustness of weak Schur sampling, setting up the problem, and coming up with the main ideas behind the proofs are original and fully human contributions. AI was primarily used to assist with the technicalities in the proofs, as well as to do numerical and computational checks for testing human ideas quicker. AI was also used as a search tool, and as a helper for technical writing. The authors take full responsibility for the content and correctness of the results presented in this paper.


\printbibliography

\appendix
\section{The Keyl measurement and the Keyl formula}
\label{app:keyl}

This appendix collects, for completeness, the precision bound of \cref{lem:perfecterror}, and the parts of the standard description of Keyl's full state tomography protocol.

\subsection{Output precision of Keyl's spectrum estimation}
\label{app:Keyl_QSE_precision}

Recall \cref{lem:perfecterror}:

\begin{lemma*}[Noiseless expected error]
  Define
  \begin{equation}
    \eta(n,d)
    \;\defeq\;
    \sup_{\varsigma}\;
    \mathbb{E}^{\,\varsigma^{\otimes n}}_{\mathrm{meas}}
    \bigl[\tfrac12\|\lambda/n-\spec(\varsigma)\|_1\bigr],
  \end{equation}
  the supremum being over all states $\varsigma$ on $\C^d$.  Then
  $\eta(n,d)=O(d/\sqrt n)$.
\end{lemma*}

\begin{proof}
  The loss $L=\tfrac12\|\lambda/n-\spec(\varsigma)\|_1$ takes values in
  $[0,1]$, so $\mathbb{E}[L]=\int_0^1\Pr[L>t]\,\mathrm{d}t$.  Splitting the
  integral at $\delta$ and inserting \cref{eq:precision},
    \begin{equation*}
    \mathbb{E}[L]
    \;\le\;\delta+\int_\delta^\infty a(n,d)\,e^{-2nt^2}\,\mathrm{d}t
    \;\le\;\delta+\frac{a(n,d)}{4n\delta}\,e^{-2n\delta^2}.
  \end{equation*}
  With $\delta=\sqrt{C\,d^2\log(n+1)/n}$ for a sufficiently large absolute
  constant $C$, the polynomial prefactor from \cref{thm:kw}, $a(n,d)=(n+1)^{O(d^2)}$, is absorbed
  by the exponential. The second term is $O(\delta)$, so
  $\eta(n,d)=O(d\sqrt{\log n/n})$ uniformly in $\varsigma$.  The logarithm
  is removed by the moment bounds of
  \cite{odonnell2016efficient,haah2017sample}, which give
  $\mathbb{E}^{\,\varsigma^{\otimes n}}_{\mathrm{meas}}\|\lambda/n-\spec(\varsigma)\|_2^2=O(d/n)$
  and hence, by using $\|\cdot\|_1\le\sqrt d\|\cdot\|_2$ from Cauchy--Schwarz, and then by using Jensen's inequality, we get
  $\eta(n,d)=O(d/\sqrt n)$.
\end{proof}

\subsection{The eigenbasis measurement}
\label{app:keylpovm}

After weak Schur sampling yields $\lambda$, the input $\rho^{\otimes n}$
collapses on the $V_\lambda$ register to the normalised state
$s_\lambda(\pp)^{-1}\pi_\lambda(\rho)$, where $\pi_\lambda$ denotes the
irreducible representation of $\GL(d,\C)$ on $V_\lambda$ and $s_\lambda$ is
the Schur polynomial \cref{eq:keylformula} below.  Keyl's protocol
\cite{keyl2006quantum} then performs a \emph{covariant} continuous POVM on
$V_\lambda$, with outcomes labelled by unitaries $U\in U(d)$, given for any
Borel set $F\subseteq U(d)$ by
\begin{equation}
  M(F)\;=\;\dim(V_\lambda)\int_F
    \pi_\lambda(U)\,\ketbra{T_\lambda}{T_\lambda}\,\pi_\lambda(U)^\dagger\,\mathrm{d}U,
  \label{eq:keylpovm}
\end{equation}
where $\ket{T_\lambda}\in V_\lambda$ is the highest-weight vector and
$\mathrm{d}U$ is the normalised Haar measure \cite{odonnell2016efficient}.
The factor $\dim(V_\lambda)$ ensures $\int_{U(d)}\mathrm{d}M(U)=\id_{V_\lambda}$
by Schur orthogonality, and the POVM is covariant,
$M(VFV^\dagger)=\pi_\lambda(V)M(F)\pi_\lambda(V)^\dagger$ for all $V\in U(d)$,
which is the structural property behind the unbiasedness of the eigenbasis
estimate \cite{holevo1982probabilistic}.  The full estimator is
$(\hat\pp,\hat U)=(\lambda/n,U)$ and the reconstructed state is
$\hat\rho=\hat U\diag(\hat\pp)\hat U^\dagger$; the highest-weight vector plays
the role of a reference frame, being the unique (up to phase) vector in
$V_\lambda$ stabilised by the diagonal subgroup, so that probing it through
all Haar-rotated frames $\pi_\lambda(U)\ket{T_\lambda}$ extracts the
eigenbasis information that weak Schur sampling discards
\cite{keyl2006quantum}.  The optimal $O(d^2/n)$ infidelity rate for full
tomography via this construction is established in
\cite{odonnell2016efficient,haah2017sample}.  Marginalising
\cref{eq:keylpovm} over $U$ returns $\id_{V_\lambda}$, so the law of
$\lambda$ is unaffected by the second measurement, as claimed in
\cref{rem:fullkeyl}.

\subsection{The Keyl formula}
\label{app:keylformula}

For $n$ identical copies of a state with ordered eigenvalues
$\pp=(p_1,\dots,p_d)$, the outcome law \cref{eq:outcomelaw} is
\begin{equation}
  P(\lambda;\rho^{\otimes n})\;=\;m_\lambda\cdot s_\lambda(p_1,\dots,p_d),
  \label{eq:keylformula}
\end{equation}
where $m_\lambda=\dim M_\lambda$ is the dimension of the irreducible
representation of $S_n$ in the block $\lambda$, given by the hook-length
formula
\begin{equation}
  m_\lambda=\dim M_\lambda=\frac{n!}{\prod_{(i,j)\in\lambda}h(i,j)},
  \label{eq:hook}
\end{equation}
with $h(i,j)=(\lambda_i-j)+(\lambda^\top_j-i)+1$ the hook length of the cell
$(i,j)$ \cite{fulton1991representation,harrow2005applications}, and
$s_\lambda(\pp)$ is the \emph{Schur polynomial}, the character of $V_\lambda$
at $\diag(p_1,\dots,p_d)$.  The derivation from Schur--Weyl duality is short:
since $\rho^{\otimes n}$ is permutation-invariant, Schur's lemma for the
$S_n$-action gives that its restriction to $V_\lambda\otimes M_\lambda$ has
the form $\varsigma_\lambda\otimes\id_{M_\lambda}$, so
$P(\lambda)=\tr[\Pi_\lambda\rho^{\otimes n}]=m_\lambda\tr_{V_\lambda}[\varsigma_\lambda]$,
and $\tr_{V_\lambda}[\varsigma_\lambda]=s_\lambda(\pp)$ is the definition of
the character \cite{fulton1991representation}.  Normalisation is the
identity $\sum_\lambda m_\lambda s_\lambda(\pp)=(\sum_ip_i)^n=1$.  The Schur
polynomial can be evaluated by the \emph{Jacobi--Trudi identity}
\begin{equation}
  s_\lambda(\pp)\;=\;\det\bigl(h_{\lambda_i-i+j}(\pp)\bigr)_{1\le i,j\le\ell(\lambda)},
  \label{eq:jt}
\end{equation}
where $h_k(\pp)=\sum_{|\alpha|=k}\pp^\alpha$ is the $k$-th complete
homogeneous symmetric polynomial (with $h_0=1$ and $h_k=0$ for $k<0$),
generated by $\prod_i(1-p_it)^{-1}=\sum_kh_k(\pp)t^k$.

For a product input $\bigotimes_i\sigma_i$ of non-identical copies there is no
analogue of \cref{eq:keylformula}: the input is not permutation-invariant,
so Schur's lemma does not apply blockwise.  By \cref{eq:symmetrisation} the
outcome law equals that of the symmetrised input
$\tfrac1{n!}\sum_\pi\pi(\bigotimes_i\sigma_i)\pi^\dagger$, which \emph{is}
permutation-invariant, but its blocks $\varsigma_\lambda$ are not characters
of a single operator, and the main text is arranged so as never to need
them.

\section{Single-shot guarantees, confidence amplification, and concentration}
\label{app:concentration}

\Cref{thm:robust,thm:wasserstein} bound an expectation.  This appendix
collects the statements that an expectation bound does not by itself give:
what a single run of the protocol guarantees and how many repetitions give a
prescribed confidence (\cref{app:singleshot}); how heavy the tails of the
outcome distribution are for a fixed input
(\cref{app:tails}); and, when the noise is itself random, how much the
centre of the outcome distribution fluctuates from one realisation to the
next (\cref{app:perrealisation}).  None of this is needed for
\cref{def:robust}; all of it follows from the same single-register bounds,
\cref{lem:centering,lem:spread,lem:qlipschitz}.  Throughout, the error (loss) of QSE is
$L(\lambda)=\tfrac12\|\lambda/n-\pp\|_1$, and we also shorten $\hat\pp \defeq \lambda/n$ for the QSE estimator.

\subsection{Single-shot guarantees and confidence amplification}
\label{app:singleshot}

Write
\begin{equation}
  E^\star\defeq\eps+\eta(n,d)
  \label{eq:Estar}
\end{equation}
for the bound of \cref{thm:robust}; for a general input one may take
$E^\star=\eps_W(\omega)+\eta(n,d) = \tfrac1n\|\omega-\rho^{\otimes n}\|_{W_1} + \eta(n,d)$ from \cref{thm:wasserstein} instead, and
everything below goes through verbatim. Here $E^\star$ bounds the \emph{expected} error; the corollaries below convert it to a \emph{precision} $\delta$ guaranteed with a prescribed \textit{confidence} $1 - \gamma$ (i.e.\ failure probability~$\gamma$), at the cost of constant factors. Because the loss $L$ takes values in
$[0,1]$ and $E^\star$ is a \emph{worst-case} bound, Markov's inequality
applies uniformly over admissible inputs.

\begin{corollary}[Single-shot guarantee]
  \label{cor:singleshot}
  For every admissible input, every worst-case expected error $E^\star$, and every precision $\delta > 0$,
  $\Pr_{\mathrm{meas}}\bigl[\tfrac12\|\hat\pp-\pp\|_1>\delta\bigr]\le E^\star/\delta$.
  In particular a single run returns an estimate within
  $3E^\star=3\eps+3\eta(n,d)$ of the true spectrum with probability at least
  $2/3$.
\end{corollary}

To reach confidence $1-\gamma$ one may sharpen the tail bound
(\cref{app:tails}) or simply repeat the experiment.  Repetition requires
nothing beyond \cref{cor:singleshot}; because the estimate is a vector, the
usual median trick is implemented by a selection rule.

\begin{lemma}[Median-of-estimates selection]
  \label{lem:median}
  Let $\hat\pp^{(1)},\dots,\hat\pp^{(k)}$ be estimates, $k$ odd, and suppose
  at least $\lceil k/2\rceil$ of them satisfy
  $\tfrac12\|\hat\pp^{(l)}-\pp\|_1\le r$.  Let $\hat\pp^\star$ be any
  $\hat\pp^{(l)}$ for which at least $\lceil k/2\rceil$ of the estimates
  (including itself) lie within $\tfrac12\|\cdot\|_1$-distance $2r$ of it.
  Then such an $l$ exists and $\tfrac12\|\hat\pp^\star-\pp\|_1\le3r$.
\end{lemma}

\begin{proof}
  Let $G$ be the set of good estimates, $|G|\ge\lceil k/2\rceil$.  Any two
  elements of $G$ are within $2r$ of each other by the triangle inequality, so
  every $l\in G$ satisfies the selection criterion, and such an $l$ exists.
  Let $l$ be selected and $N$ its set of neighbours, $|N|\ge\lceil k/2\rceil$.
  Since $k$ is odd, $|G|+|N|\ge k+1>k$, so $G\cap N\neq\emptyset$; pick
  $g\in G\cap N$.  Then
  $\tfrac12\|\hat\pp^{(l)}-\pp\|_1\le\tfrac12\|\hat\pp^{(l)}-\hat\pp^{(g)}\|_1
  +\tfrac12\|\hat\pp^{(g)}-\pp\|_1\le2r+r=3r$.
\end{proof}

\begin{corollary}[Confidence amplification by repetition]
  \label{cor:amplification}
  Run the protocol $k$ times independently, each on a fresh batch of $n$
  copies (with arbitrary admissible noise on each batch), and let
  $\hat\pp^\star$ be selected as in \cref{lem:median} with $r=3E^\star$.
  Then for $k=O(\log(1/\gamma))$ odd,
  \begin{equation*}
    \Pr\bigl[\tfrac12\|\hat\pp^\star-\pp\|_1>9E^\star\bigr]\;\le\;\gamma.
  \end{equation*}
  The total number of copies consumed is $nk=O(n\log(1/\gamma))$.
\end{corollary}

\begin{proof}
  By \cref{cor:singleshot} each run is good (within $r=3E^\star$) with
  probability at least $2/3$, independently across runs.  By the Chernoff
  bound the probability that fewer than $\lceil k/2\rceil$ runs are good is
  at most $e^{-ck}$ for an absolute constant $c>0$; choose
  $k=O(\log(1/\gamma))$.  On the complementary event \cref{lem:median}
  applies with $r=3E^\star$.
\end{proof}

\begin{remark}[Sample complexity in the presence of noise]
  \label{rem:samplecomplexity}
  Combining \cref{thm:robust} with \cref{cor:amplification}: to output an
  estimate within $\tfrac12\|\cdot\|_1$-distance $\delta$ of $\pp$ with
  confidence $1-\gamma$ it suffices to take $n=O\bigl(d^2/(\delta/9-\eps)^2\bigr)$
  copies per run and $k=O(\log(1/\gamma))$ runs, provided $\delta>9\eps$.
  The constants are not optimised; the structural content is that the copy
  cost retains the collective $d^2$ scaling and the confidence cost is
  logarithmic.
\end{remark}

\subsection{Tail bounds for a fixed input}
\label{app:tails}

We give two tail bounds for the outcome distribution of a fixed input.  The
first has exponential dependence on $n$ but an additive $\eps/w$ term; the
second is only polynomial in $n$ but has no such term and holds for every
product input, with no closeness promise at all.

\paragraph{\textit{A smoothed telescoping bound.}}
Applying \cref{prop:telescope} to the indicator function of the event
$\{L(\lambda)>\delta\}$, i.e.\ to $\mathbf{1}[L(\lambda)>\delta]\in\{0,1\}$,
is useless: an indicator jumps by $1$ across the boundary of the event and
so has one-box variation $\Theta(1)$ rather than $O(1/n)$.  Replacing the
indicator by a \emph{smoothed} version, which ramps up over a window of
width $w$, restores a small one-box variation at the price of blurring the
threshold by $w$.  We use the elementary fact that if
$\phi\colon\R\to\R$ satisfies $|\phi(x)-\phi(y)|\le k|x-y|$, then
$\Delta(\phi\circ f)\le k\,\Delta(f)$ for every $f$: for partitions
$\lambda,\lambda'$ differing by one box,
\begin{equation}
  \label{eq:one-box_composition_bound}
  \bigl|\phi(f(\lambda))-\phi(f(\lambda'))\bigr|
  \;\le\;k\,\bigl|f(\lambda)-f(\lambda')\bigr|
  \;\le\;k\,\Delta(f),
\end{equation}
and taking the maximum over such pairs gives the claim, directly from
\cref{eq:oneboxvar}.

\begin{proposition}[Deviation bound for a fixed product input]
  \label{prop:softtail}
  Let $\omega=\bigotimes_i\sigma_i$ be an admissible product input.  Then for
  every $w\in(0,1]$ and every $\delta>0$,
  \begin{equation}
    \Pr^{\,\omega}_{\mathrm{meas}}\bigl[L(\lambda)>\delta+w\bigr]
    \;\le\;a(n,d)\,e^{-2n{\delta}^2}+\frac{\eps}{w},
    \label{eq:softtail}
  \end{equation}
  with $a$ the polynomial prefactor of \cref{thm:kw}.  For a general input $\omega$
  the same holds with $\eps$ replaced by $\eps_W(\omega)$.
\end{proposition}

\begin{proof}
  Let $\phi\colon\R\to[0,1]$ be the piecewise-linear ramp
  \begin{equation}
  \label{eq:linear_ramp}
    \phi(x)\;\defeq\;\min\bigl\{1,\;\max\{0,\,(x-\delta)/w\}\bigr\},
  \end{equation}
  equal to $0$ below
  $\delta$, to $1$ above $\delta+w$, and linear in between; it satisfies
  $|\phi(x)-\phi(y)|\le|x-y|/w$.  Put $f=\phi\circ L$; then
  $\Delta(f)\le\Delta(L)/w\le1/(nw)$ by \cref{eq:one-box_composition_bound} and \cref{lem:onebox}(ii).  Since
  $\mathbf{1}[L>\delta+w]\le f\le\mathbf{1}[L>\delta]$ pointwise,
  \begin{equation*}
    \Pr^{\,\omega}_{\mathrm{meas}}[L>\delta+w]
    \;\le\;\mathbb{E}^{\,\omega}_{\mathrm{meas}}[f]
    \;\le\;\mathbb{E}^{\,\rho^{\otimes n}}_{\mathrm{meas}}[f]+n\eps\cdot\frac1{nw}
    \;\le\;\Pr^{\,\rho^{\otimes n}}_{\mathrm{meas}}[L>\delta]+\frac{\eps}{w},
  \end{equation*}
  using \cref{prop:telescope} in the middle step (or \cref{thm:wasserstein}
  for a general input), and
  $\Pr^{\,\rho^{\otimes n}}_{\mathrm{meas}}[L>\delta]\le a e^{-2n\delta^2}$ by
  \cref{thm:kw}.
\end{proof}

\paragraph{\textit{A variance bound from the quantum Efron--Stein inequality.}}
De~Palma, Fanizza, Mowry and O'Donnell \cite{depalma2025noniid} prove a
quantum version of the Efron--Stein inequality for product states, which we
now recall.  Let $\varrho=\rho_1\otimes\cdots\otimes\rho_n$ be a product state
on $(\C^d)^{\otimes n}$ and $X$ a Hermitian operator.  For $i\in[n]$ define
the linear map that averages register $i$ against $\rho_i$,
\begin{equation}
  \mathcal{E}_iX\;\defeq\;\tr_i[\rho_i\,X]\otimes_i\id,
  \qquad
  \mathcal{D}_iX\;\defeq\;X-\mathcal{E}_iX ,
  \label{eq:ESmaps}
\end{equation}
where $\tr_i[\rho_iX]$ denotes the partial trace over register $i$ of
$(\id\otimes_i\rho_i)X$, and the result is tensored with the identity on
register $i$ so that $\mathcal{E}_iX$ acts trivially there.  Write
$\operatorname{Var}_\varrho[X]\defeq\tr[\varrho X^2]-\tr[\varrho X]^2$.

\begin{theorem}[Quantum Efron--Stein inequality, {\cite[Thm.~6.6]{depalma2025noniid}}]
  \label{thm:qES}
  For every product state $\varrho$ and every Hermitian $X$,
  \begin{equation}
    \operatorname{Var}_\varrho[X]\;\le\;\sum_{i=1}^n\tr\bigl[\varrho\,(\mathcal{D}_iX)^2\bigr].
    \label{eq:qES}
  \end{equation}
\end{theorem}

For $X=\mathcal{O}_f$ the left-hand side is the variance of the measurement
outcome: since the $\Pi_\lambda$ are orthogonal projectors,
$\mathcal{O}_f^2=\mathcal{O}_{f^2}$, so
$\operatorname{Var}_\omega[\mathcal{O}_f]=\operatorname{Var}^{\,\omega}_{\mathrm{meas}}(f(\lambda))$.

\begin{proposition}[Variance of the outcome under any product input]
  \label{prop:ESvariance}
  For every product state $\omega=\bigotimes_i\sigma_i$ --- with no assumption
  on the $\sigma_i$ --- and every function $f$,
  \begin{equation}
    \operatorname{Var}^{\,\omega}_{\mathrm{meas}}\bigl(f(\lambda)\bigr)\;\le\;n\,\Delta(f)^2 .
    \label{eq:ESvariance}
  \end{equation}
  In particular $\operatorname{Var}^{\,\omega}_{\mathrm{meas}}(\lambda_k/n)\le1/n$
  for each $k$ and $\operatorname{Var}^{\,\omega}_{\mathrm{meas}}(L)\le1/n$.
  Consequently, for an admissible product input and every $\delta>0$,
  \begin{equation}
    \Pr^{\,\omega}_{\mathrm{meas}}\bigl[L(\lambda)>\eps+\eta(n,d)+\delta\bigr]
    \;\le\;\frac{1}{n\,\delta^2},
    \label{eq:chebyshev}
  \end{equation}
  and for a general product input the same holds with $\eps$ replaced by
  $\tfrac1n\sum_i\TD(\sigma_i,\rho)$.
\end{proposition}

\begin{proof}
  Fix $i$.  By \cref{eq:Ofdecomp} (transported to register $i$ by
  permutation symmetry) we can write $\mathcal{O}_f=H'_i+B_i$ with $H'_i$
  acting trivially on register $i$ and $\|B_i\|_{\mathrm{op}}\le\Delta(f)/2$.
  Since $H'_i$ acts trivially on register $i$, $\mathcal{E}_iH'_i=H'_i$ and
  $\mathcal{D}_iH'_i=0$, so $\mathcal{D}_i\mathcal{O}_f=\mathcal{D}_iB_i
  =B_i-\tr_i[\sigma_iB_i]\otimes_i\id$.  The map $B\mapsto\tr_i[\sigma_iB]$
  is a convex combination of compressions and does not increase the operator
  norm, so $\|\mathcal{D}_i\mathcal{O}_f\|_{\mathrm{op}}\le2\|B_i\|_{\mathrm{op}}\le\Delta(f)$
  and $\tr[\omega(\mathcal{D}_i\mathcal{O}_f)^2]\le\Delta(f)^2$.  Summing over
  $i$ in \cref{eq:qES} gives \cref{eq:ESvariance}; the special cases are
  \cref{lem:onebox}.  For \cref{eq:chebyshev}, Chebyshev's inequality gives
  $\Pr^{\,\omega}_{\mathrm{meas}}[L>\mathbb{E}^{\,\omega}_{\mathrm{meas}}[L]+\delta]\le\operatorname{Var}^{\,\omega}_{\mathrm{meas}}(L)/\delta^2\le1/(n\delta^2)$,
  and $\mathbb{E}^{\,\omega}_{\mathrm{meas}}[L]\le\eps+\eta(n,d)$ by
  \cref{thm:robust} (or \cref{cor:w1cases}(i)).
\end{proof}

\begin{remark}
  \label{rem:ESreading}
  \Cref{prop:ESvariance} says that the outcome of weak Schur sampling on
  \emph{any} product input has variance $O(1/n)$ per coordinate, exactly as
  for the noiseless input, where $\operatorname{Var}(\lambda_k/n)\approx
  p_k(1-p_k)/n$ \cite{keyl2001estimating,odonnell2016efficient}.  The noise
  therefore affects only \emph{where} the outcome distribution is centred,
  which is what \cref{thm:robust,thm:wasserstein} control, and not how
  concentrated it is.  Sharpening \cref{eq:chebyshev} to an exponential tail
  without the $\eps/w$ term of \cref{eq:softtail} would require a
  concentration inequality for Lipschitz observables on product states; see
  item~(4) of \cref{sec:open}.
\end{remark}

\paragraph{\textit{Unbiased noise.}}
For randomly drawn unbiased noise (\cref{def:unbiasedbiased}) the transfer of
the noiseless tail bound is exact and needs no smoothing.

\begin{proposition}[Unbiased noise: exact tail identity]
  \label{prop:tail}
  Under unbiased noise, for every $\delta>0$,
  \begin{equation*}
    \mathbb{E}_{\mathrm{noise}}\Bigl[\Pr^{\,\rho_1\otimes\cdots\otimes\rho_n}_{\mathrm{meas}}[L>\delta]\Bigr]
    =\Pr^{\,\rho^{\otimes n}}_{\mathrm{meas}}[L>\delta]
    \le a(n,d)\,e^{-2n\delta^2},
  \end{equation*}
  and consequently, by Markov's inequality, for every $s>0$ the conditional
  tail probability of the realised input exceeds $s$ with noise-probability
  at most $a\,e^{-2n\delta^2}/s$.
\end{proposition}

\begin{proof}
  $\Pr^{\,\omega}_{\mathrm{meas}}[L>\delta]=\tr[\Pi_{>\delta}\,\omega]$ with
  $\Pi_{>\delta}\defeq\sum_{\lambda:L(\lambda)>\delta}\Pi_\lambda$ is linear in
  $\omega$, so \cref{eq:meanidentity} with $\rho_i^E=\rho$ gives the
  identity; then apply \cref{thm:kw}.
\end{proof}

\subsection{Concentration of the conditional centre}
\label{app:perrealisation}

\Cref{cor:unbiasedrobust} says that under unbiased noise the performance is
noiseless \emph{on average over the noise}.  A stronger and more practically
relevant statement is that the noise realisation one actually receives is
almost surely a good one.  Write
\begin{equation}
  \hat\mu_j
  \defeq\mathbb{E}^{\,\rho_1\otimes\cdots\otimes\rho_n}_{\mathrm{meas}}\bigl[\lambda_j/n\bigr]
  =\tr\bigl[\mathcal{O}_{\lambda_j/n}\,(\rho_1\otimes\cdots\otimes\rho_n)\bigr]
  \label{eq:condcentre}
\end{equation}
for the \emph{conditional centre}: the mean of the outcome distribution
produced by the particular input that the noise happened to deliver.  It is a
random vector through the $\rho_i$ only.  The tool is the classical
Efron--Stein inequality, now applied to the noise randomness rather than to
the measurement.

\begin{theorem}[Efron--Stein inequality, {\cite[Thm.~3.1]{boucheron2013concentration}}]
  \label{thm:es}
  Let $X_1,\dots,X_n$ be independent random variables, let
  $g(X_1,\dots,X_n)$ be square-integrable, and let $X_i'$ be an independent
  copy of $X_i$.  Writing
  $g^{(i)}\defeq g(X_1,\dots,X_{i-1},X_i',X_{i+1},\dots,X_n)$,
  \begin{equation*}
    \operatorname{Var}(g)\;\le\;\frac12\sum_{i=1}^n\mathbb{E}\bigl[(g-g^{(i)})^2\bigr].
  \end{equation*}
\end{theorem}

\begin{theorem}[Concentration of the conditional centre]
  \label{thm:centre}
  Under randomly drawn admissible product noise, for each $j\in[d]$,
  \begin{equation*}
    \operatorname{Var}_{\mathrm{noise}}(\hat\mu_j)\;\le\;\frac{2\eps^2}{n},
    \qquad\text{hence}\qquad
    \mathbb{E}_{\mathrm{noise}}\bigl\|\hat\mu-\mathbb{E}_{\mathrm{noise}}[\hat\mu]\bigr\|_2^2
    \;\le\;\frac{2d\,\eps^2}{n}.
  \end{equation*}
  Consequently, for every $\gamma\in(0,1)$, with probability at least
  $1-\gamma$ over the noise,
  $\|\hat\mu-\mathbb{E}_{\mathrm{noise}}[\hat\mu]\|_2\le\sqrt{2d/\gamma}\;\eps/\sqrt n$.
  Moreover, for every $s>0$,
  \begin{equation}
    \Pr_{\mathrm{noise}}\bigl[\,|\hat\mu_j-\mathbb{E}_{\mathrm{noise}}[\hat\mu_j]|>s\,\bigr]
    \;\le\;2\exp\!\Bigl(-\frac{n\,s^2}{2\eps^2}\Bigr).
    \label{eq:mcdiarmid}
  \end{equation}
\end{theorem}

\begin{proof}
  Resampling register $i$ replaces $\rho_i$ by an independent copy $\rho_i'$
  and changes only that slot of the product, so by \cref{eq:hi},
  \begin{equation*}
    \hat\mu_j-\hat\mu_j^{(i)}
    =\tr\bigl[(\rho_i-\rho_i')\,h^{\lambda_j/n}_i(R_{\neq i})\bigr],
    \qquad R_{\neq i}\defeq\textstyle\bigotimes_{k\neq i}\rho_k .
  \end{equation*}
  The operator $\rho_i-\rho_i'$ is traceless with
  $\|\rho_i-\rho_i'\|_1\le\|\rho_i-\rho\|_1+\|\rho-\rho_i'\|_1\le4\eps$, so
  \cref{lem:centering,lem:spread} (with $\Delta(\lambda_j/n)\le1/n$) give
  \begin{equation}
    \bigl|\hat\mu_j-\hat\mu_j^{(i)}\bigr|
    \;\le\;\tfrac12\cdot4\eps\cdot\tfrac1n
    \;=\;\frac{2\eps}{n}
    \qquad\text{almost surely.}
    \label{eq:swapbound}
  \end{equation}
  Substituting into \cref{thm:es},
  $\operatorname{Var}(\hat\mu_j)\le\tfrac12\cdot n\cdot(2\eps/n)^2=2\eps^2/n$.
  Summing over $j$ gives the $\ell^2$ statement, and Markov's inequality
  applied to $\|\hat\mu-\mathbb{E}\hat\mu\|_2^2$ gives the high-probability bound:
  \begin{equation}
  \label{eq:high_prob_bound}
      \|\hat\mu-\mathbb{E}_{\mathrm{noise}}[\hat\mu]\|_2\le\sqrt{2d/\gamma}\;\eps/\sqrt n
  \end{equation}
  with probability $(1-\gamma)$ over the noise.  For \cref{eq:mcdiarmid}, \cref{eq:swapbound} is a
  bounded-differences condition with constants $c_i=2\eps/n$, so
  $\sum_ic_i^2=4\eps^2/n$, and McDiarmid's inequality
  \cite[Thm.~6.2]{boucheron2013concentration} gives
  $2\exp(-2s^2/\sum_ic_i^2)$, which is \cref{eq:mcdiarmid}.
\end{proof}

\begin{remark}
  \label{rem:centreinterpretation}
  Under unbiased noise $\mathbb{E}_{\mathrm{noise}}[\hat\mu]$ is the
  noiseless centre by \cref{cor:unbiasedrobust}, so \cref{thm:centre} says
  that the centre of the outcome distribution produced by the realised input
  is within $O(\eps/\sqrt n)$ of the noiseless centre, with high probability
  over the noise.  Contrast this with the deterministic worst case of
  \cref{prop:floor}, where the centre sits a full $\eps$ away: averaging over
  independent fluctuations is what buys the extra $1/\sqrt n$.  Unbiasedness
  is not used in \cref{thm:centre} itself --- only independence, for
  \cref{thm:es}, and closeness, for \cref{eq:swapbound}.  Combining
  \cref{thm:centre} (fluctuation of the centre over the noise) with
  \cref{prop:ESvariance} (fluctuation of the outcome around the centre) gives
  a complete picture of a single run under unbiased random noise: with high
  probability over both sources of randomness the outcome is within
  $O(d/\sqrt n)$ of the true spectrum in $\tfrac12\|\cdot\|_1$, with no
  dependence on $\eps$.
\end{remark}

\end{document}